\documentclass[10pt,journal]{IEEEtran}
\usepackage{cite}
\usepackage{amsmath,amssymb,amsfonts,latexsym}
\usepackage{graphicx}
\usepackage{textcomp}
\usepackage{xcolor}
\usepackage{tabularx}
\usepackage{mathtools}
\usepackage{optidef}
\usepackage{amsthm}
\usepackage{color}
\usepackage{mathabx}
\newtheorem{theorem}{Theorem}

\usepackage{algorithm}
\usepackage{algpseudocode}
\usepackage{caption}
\usepackage{subcaption}
\usepackage{verbatim}
\usepackage{booktabs}
\usepackage{multirow}

\usepackage{dsfont}
\usepackage{bm}

\newtheorem{definition}{Definition}

\DeclareMathOperator*{\argmax}{arg\,max}
\usepackage{tikz}

\newcommand{\copyrighttext}{%
\scriptsize
\textcopyright~2026 IEEE. Personal use of this material is permitted. Permission from IEEE must be obtained for all other uses, in any current or future media, including reprinting/republishing this material for advertising or promotional purposes, creating new collective works, for resale or redistribution to servers or lists, or reuse of any copyrighted component of this work in other works.%
}

\newcommand{\copyrightnotice}{%
\begin{tikzpicture}[remember picture,overlay]
\node[anchor=south,yshift=10pt] at (current page.south) {%
\fbox{%
\parbox{\dimexpr0.90\textwidth-2\fboxsep-2\fboxrule\relax}{%
\copyrighttext
}%
}%
};
\end{tikzpicture}%
}

\begin{document}

\title{How Can Incentives and Cut Layer Selection Influence Data Contribution in Split Federated Learning?}

 \author{Joohyung~Lee,~\IEEEmembership{Senior Member,~IEEE,} Jungchan Cho, Wonjun Lee,~\IEEEmembership{Fellow,~IEEE}, \\Mohamed Seif, and H. Vincent Poor, ~\IEEEmembership{Life Fellow,~IEEE}

\thanks{This work was partly supported by the National Research Foundation of Korea (NRF) grant funded by the Korea government (MSIT) (No. RS-2024-00338786, RS-2025-22442984) and the Institute of Information and Communications Technology Planning and Evaluation (IITP) grant funded by the Korea government (MSIT) (No. RS-2023-00234719, SW Star Lab). The work of H. Vincent Poor was supported in part by a grant from Princeton Language and Intelligence.}
\thanks{J. Lee, J. Cho, and W. Lee are with the Department of Computing, Gachon University, Seongnam 13120, Rep. of Korea (e-mail: j17.lee@gachon.ac.kr; thinkai@gachon.ac.kr; wonjun.lee@gachon.ac.kr).}
\thanks{M. Seif and H. V. Poor are with the Department of Electrical and Computer Engineering, Princeton University, Princeton, NJ 08544 USA (e-mail: mseif@princeton.edu; poor@princeton.edu).}
\thanks{M. Seif is also with the Department of Computer Science and Engineering, Oakland University, MI 48309 USA (e-mail: mohamed2@oakland.edu).}

 \thanks{Corresponding author: Jungchan Cho.}}

\markboth{Journal of \LaTeX\ Class Files,~Vol.~14, No.~8, August~2021}%
{Shell \MakeLowercase{\textit{et al.}}: A Sample Article Using IEEEtran.cls for IEEE Journals}


\maketitle

\copyrightnotice

\begin{abstract}
To alleviate the training burden in federated learning while enhancing convergence speed, Split Federated Learning (SFL) has emerged as a promising approach by combining the advantages of federated and split learning. However, despite its advantages, existing SFL studies have largely overlooked the strategic interactions among self-interested participants during the SFL process. In this framework, the SFL model owner can choose the cut layer to balance the training load between the server and clients, ensuring the necessary level of privacy for the clients. Additionally, the SFL model owner sets incentives to encourage client participation in the SFL process. The optimization strategies employed by the SFL model owner influence clients' decisions regarding the amount of data they contribute, taking into account the shared incentives over clients and anticipated energy consumption from both computation and networking during SFL. To address this framework, we model the problem using a hierarchical decision-making approach, formulated as a single-leader multi-follower Stackelberg game. We demonstrate the existence and uniqueness of the Nash equilibrium among clients and analyze the Stackelberg equilibrium by examining the leader's game. Furthermore, we discuss privacy concerns related to differential privacy and the criteria for selecting the minimum required cut layer. Our findings show that the Stackelberg equilibrium solution maximizes the utility for both the clients and the SFL model owner while achieving a well-balanced trade-off between model accuracy and the associated computing and networking overhead during the SFL process.

\end{abstract}


\section{Introduction}
Recent innovations in Distributed Collaborative Machine Learning (DCML) have opened new possibilities for developing scalable and accurate Artificial Intelligence (AI) solutions. By leveraging diverse datasets across multiple clients without transmitting personal data to a central server, these methods significantly reduce the risk of privacy leakage. Federated learning (FL) and split learning (SL) are two popular approaches, though both face practical limitations \cite{Lee2024, Thapa2022,Abhishek2019}:
\begin{itemize}
    \item FL: FL allows a server model to be trained in parallel across many clients, but each client is required to run an entire model, resulting in a significant training burden. This leads to high battery consumption on clients, significant client-side computation requirements, and limited model privacy.
    \item SL: SL splits the entire model into smaller portions, allowing only a shallow sub-network to be trained on the client side, thus reducing the training burden. However, the training process is sequential, and client models are not aggregated, leading to significant time overhead.
\end{itemize}

To overcome these limitations, a new distributed learning framework, Split Federated Learning (SFL), has recently been proposed. SFL combines the advantages of both FL and SL, aiming to reduce the training burden of FL while enhancing convergence speed. More specifically, in SFL, clients train only a portion of the entire model, known as the client-side model, which decreases their computational load. These client-side models are then synchronized to enhance convergence speed via FL process \cite{Thapa2022,Houda2023}.

However, despite its notable advantages, SFL introduces additional communication overhead when interacting with servers and also raises privacy concerns due to the frequent exchange of client-side model outputs and updates, which are correlated to raw data. Optimizing SFL management to address these challenges remains a complex issue \cite{Lee2024,Thapa2022}. Hence, there have been various approaches to managing SFL in the literature, aiming to achieve efficient communication in SFL operations. These approaches consider factors such as cut layer selection, clustering, management of radio and computational resources (e.g., CPU/GPU), and gradient compression \cite{Han2021, Thapa2022, Wu2023, Yujia2023, Lin2024}.

Nevertheless, the aforementioned work assumes that clients are motivated to voluntarily participate in the SFL process without any incentives. In reality, participating in the SFL process requires substantial energy consumption, which depends on the complexity of the client-side model and the amount of data contribution to train the model from the client's perspective. Therefore, without suitable incentives that take into account the cut layer selection (i.e., the complexity of the client-side model), motivating clients to engage in the SFL process to generate an accurate model is challenging. Although there have been various studies on incentive mechanisms in FL (see, for example, \cite{Zhan2022, Hu2020, Yuan2024, Pang2023} and the references therein), none have specifically addressing SFL's unique characteristics as follows. i) Competition arises between the SFL model owner, who aims to create a global model from the SFL process for AI services, and the clients, who participate in the SFL process. Additionally, given the incentives, there is further competition among clients striving to earn more rewards. Thus, understanding the clients' reactions to the optimal decisions and behaviors of the SFL model owner is crucial. 
ii) This necessitates a careful study of sophisticated incentive mechanisms and cut layer selection, considering how these competitive behaviors can influence the data contributions of clients. For instance, if the client-side model size increases, resulting in higher model complexity, clients will face a greater training burden. This may lead them to reduce their data contributions unless they receive sufficient incentives.

Unlike conventional FL incentive mechanisms that assume a fixed learning architecture, this work develops an analytical framework for SFL that integrates incentive design, cut-layer selection, and client data contribution. The proposed framework explicitly captures how the cut-layer decision affects client computation, networking overhead, and ultimately the incentive-driven data contribution. Accordingly, our objective is to develop an integrated approach for optimizing privacy-constrained cut-layer selection and incentive management while characterizing the strategic behavior of self-interested clients. Our contributions are summarized as follows:
\begin{itemize}
    \item We investigate the strategic interactions between the SFL model owner and participating clients during the SFL process. These interactions are formulated as a single-leader, multi-follower Stackelberg game.
    \item To achieve this, we develop analytical utility models that characterize the trade-offs faced by both the SFL model owner and clients. Specifically, the SFL model owner's utility captures the balance among satisfaction from training-data acquisition, load reduction, and incentive payments, while each client's utility captures the trade-off between the expected incentives and the energy consumption incurred by both computation and networking. In contrast to conventional FL incentive models, these utility functions are explicitly coupled with the privacy-constrained cut-layer selection and clients' strategic data-contribution decisions, enabling a unified analysis of incentive management and SFL architecture design.    
    \item We derive closed-form expressions for the Nash equilibrium among heterogeneous clients and develop algorithms to obtain Stackelberg strategies by proving the existence of these solutions. Additionally, we discuss privacy enhancements related to cut layer selection and differential privacy.
    \item Our findings show that the Nash and Stackelberg strategies are desirable operating points for all participants (i.e., the SFL model owner and clients) in competitive situations, while achieving a well-balanced trade-off between model accuracy and the associated computing and networking overhead during the SFL process. From our experiments using CIFAR-10, CIFAR-100, and MPII, we evaluate how incentive-induced equilibrium data contributions translate into actual SFL model performance. Specifically, our results characterize the relationship among the incentive, the equilibrium data contribution, and the resulting learning performance across different learning tasks and parameter settings.
    \end{itemize}

The remainder of this paper is organized as follows. Section II reviews related work on cut layer selection and resource management for SFL. Section III presents the system model and formulates the utility functions for both the SFL model owner and the clients. In Section IV, we describe the competitive game between the SFL model owner and the clients, rigorously formulating the problem as a Stackelberg game. We then characterize the closed-form solutions for a Nash Equilibrium of the clients and provide the algorithm to obtain a Stackelberg Equilibrium. Section V provides the performance evaluation. Section VI discusses the impact of cut layer selection and differential privacy on the privacy level, offering guidelines for setting the minimum cut layer requirements during the SFL process. Finally, Section VII concludes the paper and suggests directions for future research.

 \begin{figure*}       
	\centering	\includegraphics[width=400pt,keepaspectratio]{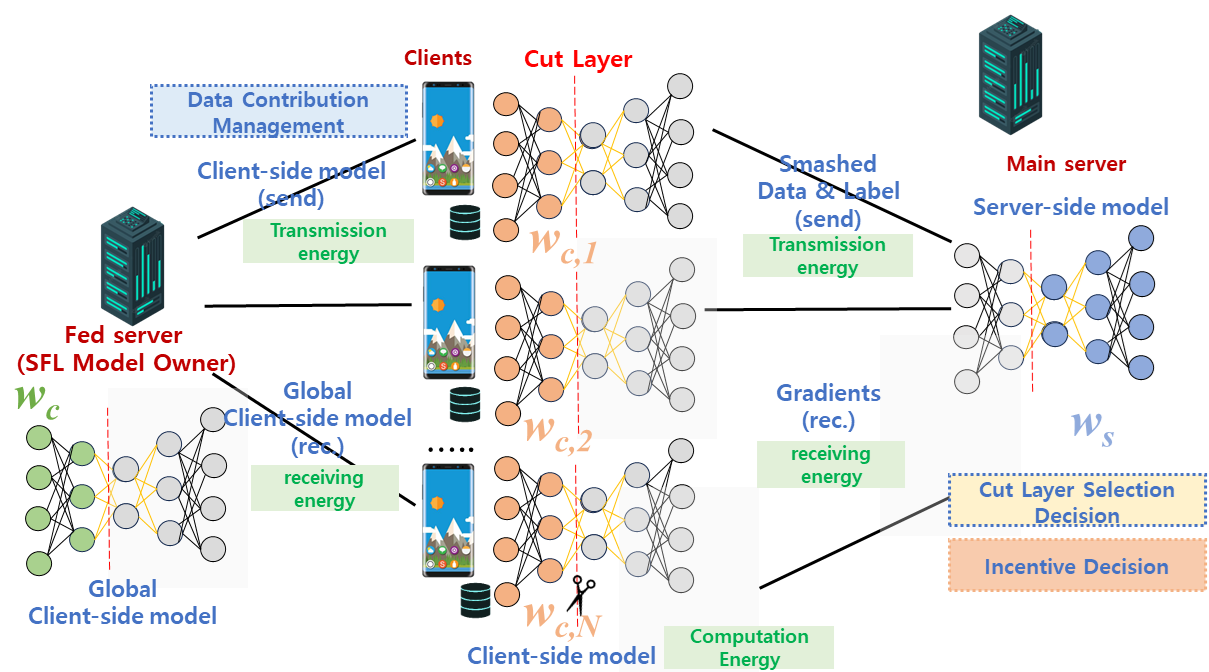}
    \caption{Concept of the proposed Split Federated Learning: It considers interactions between the model owner and clients. The SFL model owner determines the cut layer and incentives, while the clients manage data contribution, considering their energy consumption.}
    \label{fig:Figure_sfl_framework}
   \end{figure*}

\section{Related Work}
To mitigate the training burden in FL and improve convergence speed, SFL has emerged as a promising method, integrating the benefits of federated and split learning. However, SFL is associated with substantial communication overhead, mainly due to the transfer of smashed data, gradients, and model updates between clients, the main server, and the federated server. This necessitates enhancements \cite{Yujia2023}. Recent studies have focused on the impact of cut layer selection on system performance and privacy. Notably, \cite{Han2021, Thapa2022} developed analytical models to explore the effect of total model training time (latency) on the cut layer point. Based on these analyses, various strategies have been proposed to enhance communication efficiency in SFL \cite{Yujia2023, Han2021, Wu2023}. For instance, \cite{Han2021} introduced an optimal cut layer selection method based on latency analysis, followed by a convergence analysis of the SFL framework. Extending this work, \cite{Wu2023} examined the latency of SFL in wireless networks and designed a joint optimization problem involving cut layer selection, client clustering, and bandwidth allocation to minimize overall training latency. Similarly, \cite{Zhu2024} proposed a joint method for cut layer selection and server computation resource allocation for clients. \cite{Yujia2023} employed an auxiliary network for local updates of client-side models, maintaining only a single server-side model to reduce storage costs on the main server, and conducted a comprehensive analysis of communication costs, storage space, and convergence. Additionally, \cite{Lin2024} focused on reducing the dimensions of activations' gradients for backpropagation to decrease communication overhead. This study also explored joint optimization strategies for subchannel allocation, power control, and cut layer selection in wireless networks, aiming to minimize per-round latency. The recent work of \cite{Solat2024} utilized SFL for Uncrewed Aerial Vehicles (UAVs) networks considering wireless network perspectives. It generated the mobile traffic prediction model from UAVs collaboratively by selecting the cut layer in SFL for energy-efficient training.

In addition to these efficiency considerations, SFL poses privacy concerns due to interactions between clients and servers, including both the main and federated servers. During the forward phase, clients transmit smashed data to the main server, which becomes vulnerable to reconstruction attacks, potentially compromising the privacy of the original information. Furthermore, exchanging model updates with the federated server introduces additional risks to raw data privacy. To address these issues, several privacy-preserving mechanisms have been proposed. For instance, a recent study \cite{Kim2020} on split learning empirically investigated the impact of cut layer selection on reconstruction attacks during the forward phase. The findings suggested that having deeper layers on the client side, which involve more non-linear functions, compresses raw data by eliminating less informative features, thereby enhancing resistance to reconstruction attacks. Similarly, \cite{Vepakomma2019} conducted empirical studies on how the distance correlation between raw and smashed data affects privacy leakage in split learning, emphasizing the importance of cut layer selection. Recently, the work in \cite{Lee2024a} analyzed the trade-off between energy consumption and privacy levels in relation to cut-layer selection in SFL. Furthermore, the work of \cite{Lee2024b} introduced an innovative cut layer selection scheme aimed at reducing the overall training latency in SFL, taking into account wireless network perspectives. Specifically, the proposed scheme not only minimized the overall training latency but also ensured an acceptable level of privacy by setting boundary constraints on the cut layer selection.

However, none of these studies address the strategic mechanisms of the SFL model owner and clients considering their competitive behaviors, which is one of our contributions. Notably, our approach to formulating and analyzing a Stackelberg game with two solution concepts involving the SFL model owner and clients is not fully covered in previous work.

 \begin{figure}       
	\centering	\includegraphics[width=260pt,keepaspectratio]{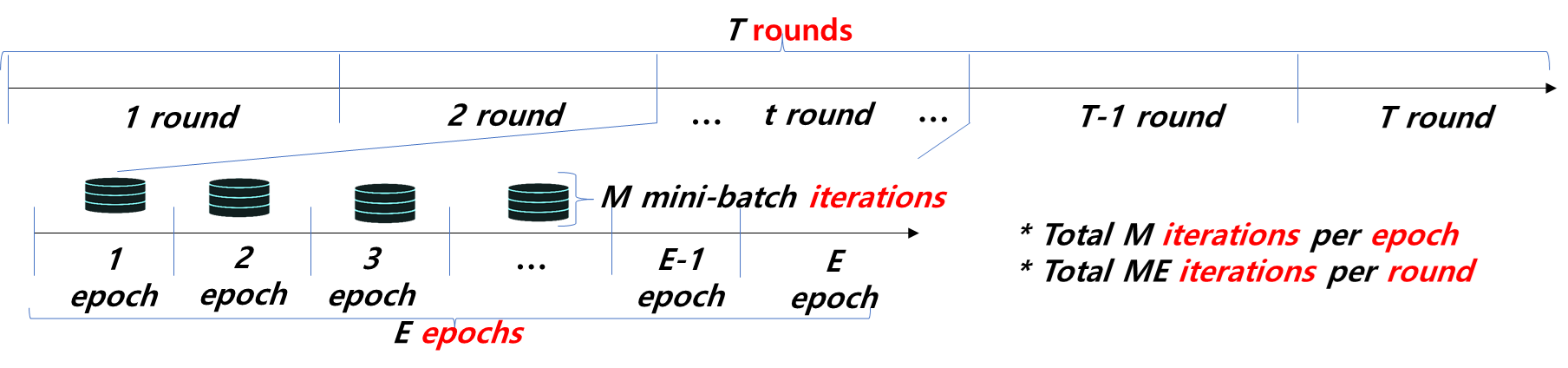}
    \caption{Timeline of the Overall SFL Process.}
    \label{fig:Timeline_of_SFL}
   \end{figure}

\section{System Model}
The SFL framework consists of two servers 
: (i) a federated learning server, referred to as the fed server, and (ii) a main server, along with multiple clients participating as part of the SFL process. In this setup, we define a set of clients as $\mathcal{N} = {1, 2, \cdots, N}$, where each client is represented as $n \in \mathcal{N}$ and has its own local dataset $\mathcal{D}_n$. The size of each dataset $\mathcal{D}_n$ is indicated by the number of data items, denoted as $D_n = |\mathcal{D}_n|$, with each data item comprising both features and labels. The entire model, denoted as $w$, is composed of $L$ layers. Specifically, the model $w$ is divided into two sub-models at the $L_c$-th layer, referred to as the cut layer. These sub-models are: (i) the client-side model $w_c$, including layers from the first layer up to layer $L_c$, and (ii) the server-side model $w_s$, encompassing layers from layer $L_c+1$ to layer $L$. Thus, the full model can be represented as $w = [w_c; w_s]$. As shown in Fig. \ref{fig:Figure_sfl_framework}, each client $n$ in the set $\mathcal{N}$ maintains its client-side model, denoted as $w_{c,n}$. These client-side models are trained in parallel using the clients' respective local datasets, through interactions with the server-side model $w_s$ hosted on the main server. To synchronize the client-side models, clients participate in federated learning, coordinated by the fed server, to construct a global client-side model $w_c$ in a collaborative manner \cite{Huang2023}. In our scenario, the SFL model owner is deployed on top of the fed server and determines the appropriate cut layer, as well as the incentives that motivate clients to participate in the SFL process. It should be noted that the fed server and the main server can be either physical or logical entities. Accordingly, both servers may also be implemented on the same physical machine.

Finally, the entire SFL process is carried out over $T$ rounds. At the beginning of each round $t$, clients are provided with the latest global client-side model $w_c$ from the fed server. Following this, both the clients and the main server engage in training over $E$ epochs, updating their respective client-side models $w_{c,n}$ for each client $n \in \mathcal{N}$, and the server-side model $w_s$. Then, each client $n$ randomly samples a mini-batch instance $\zeta_n$ from its dataset $\mathcal{D}_n$. Note that during each epoch, clients randomly partition their datasets into $M$ mini-batches. Stochastic gradient descent (SGD) is applied to each mini-batch, allowing clients to update their client-side models $w_{c,n}$ and the main server to update the server-side model $w_s$ in a coordinated fashion. Once the $E$ epochs are concluded, the fed server collects the updated client-side models $w_{c,n}$ from all clients and combines them to form an updated global client-side model $w_c$\footnote{At mini-batch instance $\zeta_n$, $l_n(w; \zeta_n)$ denotes the loss of the model $w$ over client $n$'s mini-batch instance $\zeta_n$, which can be expressed as follows:
\begin{equation}\label{eqn:f_n(w,zeta)}
l_n(w; \zeta_n) = h(w_s; u(w_{c,n}; \zeta_n)).
\end{equation}
Here, $u$ is a function on the client side that transforms the input data from the sampled mini-batch instance $\zeta_n$ into the activation space, producing the output of the client-side model $w_{c,n}$ (i.e., smashed data). Meanwhile, $h$ is a function on the server side that takes the activation space as input and produces a single scalar value representing the loss.}. A detailed timeline of overall SFL process is summarized in Fig. \ref{fig:Timeline_of_SFL}.

\subsection{Utility Function of SFL Model Owner}
In our scenario, the SFL model owner determines the cut layer $L_c$ and the incentives $R$ provided to clients based on their data contributions, with the goal of maximizing the owner's utility function. In this context, we consider the following utility function:
\begin{align} \label{eqn:utility_MO}
  U_{MO}(R, L_c, \textbf{d})=\tau_1 \ln\left(1+\frac{\sum_{n = 1}^N d_n}{d_{req}}\right)+\tau_2 \frac{f_{FLOPs}(L_c)}{w_{FLOPs}} - R,
\end{align}
where the vector $\mathbf{d} = (d_1, d_2, \dots, d_N)$ contains the data contribution of each client from its local dataset $\mathcal{D}_n$, while $d_{req}$ denotes the required amount of aggregated data contribution necessary to achieve an acceptable accuracy, which is a service-dependent requirement. $\tau_1, \tau_2\ge 0$ are weights for trading of the first and second terms. Moreover, $w_{FLOPs}$ is the computation workload to proceed with the full model $w$ during its forward and backward propagation for a single data sample, measured in Floating Point Operations Per Second (FLOPS). Here, $f_{FLOPs}(L_c)$ as a function of the computation workload, measured in FLOPS, of the client-side model $w_c$ with respect to $L_c$. In this context, $f_{FLOPs}(L_c)$ is an increasing function of $L_c$ since the computation workload increases with the number of layers. Nevertheless, due to the difficulty of analytically characterizing this model-dependency, a data-driven approach can be applied to model its relationship through offline training using empirical measurements. Hence, we adopt regression-based modeling, which is one of the most widely used approaches for modeling mobile CPU properties, such as CPU power and temperature variation modeling \cite{Wang2020, Liu2018a}. Consequently, as shown in Fig. \ref{fig:FLOPs_Params}, we conclude that $f_{FLOPs}(L_c)$ is an affine function of $L_c$, modeled as $f_{FLOPs}(L_c)=aL_c + b$, where $a$ and $b$ are coefficient values\footnote{We used the thop package to measure the forward FLOPs and the model Params: https://pypi.org/project/thop/. Backword FLOPs are assumed to be twice the forward FLOPs.}. Finally, this utility function quantifies the balance between the satisfaction derived from aggregated data contribution from clients, and the gain from distributing training workload to clients, against the expenditure $R$ allocated to clients during the SFL process. The range of $R$ is set by the SFL model owner within $(0, R_{\max}]$, where $R_{\max}$ represents the upper limit of incentives regulated by the SFL model owner's policy. Note that in the right-hand side of (\ref{eqn:utility_MO}), the first term $\ln(1+\frac{\sum_{n = 1}^N d_n}{d_{req}})$ suggests that the satisfaction gained from aggregated data contribution may follow the law of diminishing returns. From \cite{Niyato_2016}, the amount of raw data plays an important role in the quality of data analytics (e.g., model accuracy). Specifically, in the early stages, adding more client data often leads to noticeable improvements in model accuracy. However, as the dataset becomes larger, each additional unit of data provides smaller performance gains. Consistent with this trend, numerous studies demonstrate that model accuracy increases concavely with the overall dataset size \cite{Luo2020, Kim2022, Doyeon2023}. Accordingly, a concave, strictly increasing, and continuously differentiable function can represent this term. In this paper, we adopt a logarithmic utility function for this purpose, as it is widely used in the literature to quantify user satisfaction with diminishing returns. 
\subsection{Utility Function of Clients}
Our design of the utility function for the clients comprises two terms. The first term denotes the incentive allocated by the SFL model owner, while the second term represents the energy expenditure incurred from participating in the SFL process.

 \begin{figure}       
	\centering	\includegraphics[width=220pt,keepaspectratio]{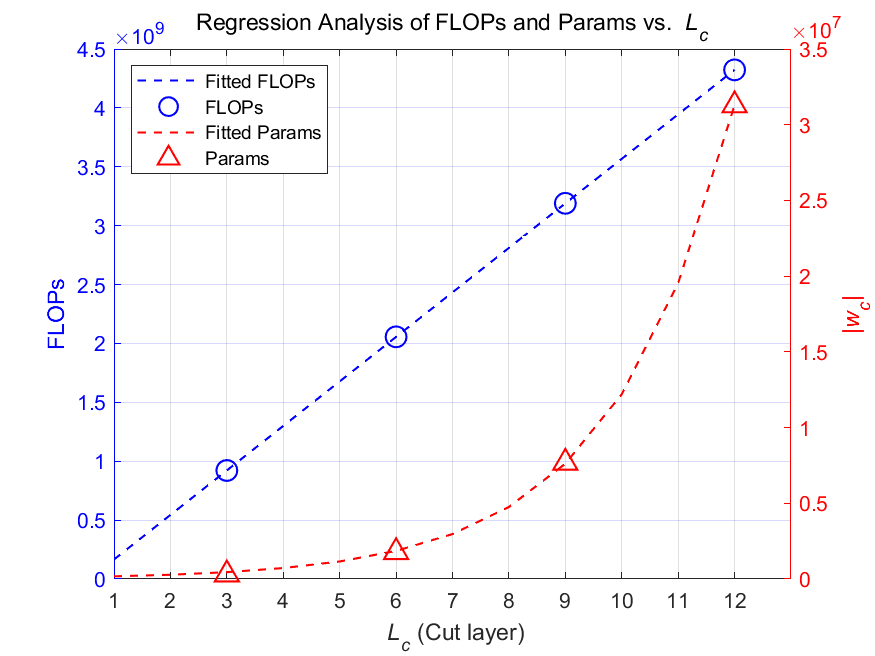}
    \caption{Regression analysis of computation load and the number of model parameters with varying $L_c$.}
    \label{fig:FLOPs_Params}
   \end{figure}

\subsubsection{Overall computation energy of the client}
Using $f_{FLOPs}(L_c)$ to model the computation time for each client $n$, we obtain computation latency for the client-side model to process the mini-batch instance $\zeta_k$ can be obtained by
\begin{equation} \label{computation_time}
   \footnotesize t_{n} = \frac{d_n f_{FLOPs}(L_c)}{Mf_n \kappa}, \forall n \in \mathcal{N},
    \end{equation}
where the size of the mini-batch instance $\zeta_k$ is denoted by $|\zeta_k| = d_n/M$. Furthermore, $f_n$ denotes the central processing unit (CPU) capability of client $n$, and $\kappa$ represents the computing intensity. Consequently, the energy consumption for computing the client-side model for client $n$ is given by $E_{cmp, n} = p_n^c t_{n}$, where $p_n^c$ is the computation power of client $n$.

\subsubsection{Overall communication energy of the client}
To streamline our analysis and focus on the core intuition, we adopt a simplified approach similar to \cite{Tran2019}. This involves assuming fixed uplink and downlink channels, maintaining a quasi-static channel mode throughout the SFL process facilitating efficient analysis. Then, transmission latencies for transmitting the smashed data from client $n$ to the main server and for receiving the gradient data from the main server to the client $n$, are given by:
\begin{equation}\label{eqn:t_main}
t_{s:n}= \frac{s_n}{r_{s:n}}, \quad  t_{g:n} = \frac{g_n}{r_{g:n}}, \quad \forall n \in \mathcal{N}.
\end{equation}
Here, $s_n$ and $g_n$ denote the size of the smashed data and the gradient data of client $n$, respectively\footnote{Here, we assume that the size of the smashed data is constant, although it could be model-dependent. For the evaluation, we set the size of the smashed data by averaging the size across layers to simplify the analysis.}. Additionally, $r_{s:n}$ and $r_{g:n}$ represent the uplink transmission rate from client $n$ to the main server and the downlink transmission rate from the main server to client $n$, respectively.
Similarly, transmission latencies for transmitting the client-side model from client $n$ to the fed server and for receiving the global client-side model from the fed server to the client $n$, are given by:
\begin{equation}\label{eqn:t_main}
t_{cm:n}= \frac{m|w_c|}{r_{cm:n}}, \quad  t_{gcm:n} = \frac{m|w_c|}{r_{gcm:n}}, \quad \forall n \in \mathcal{N}.
\end{equation}
Here, $r_{cm:n}$ and $r_{gcm:n}$ represent the uplink transmission rate from client $n$ to the fed server and the downlink transmission rate from the fed server to client $n$, respectively. $m$ represents the number of bits for a single model parameter, where $|w_c|$ is the total number of model parameters in the client-side model. Note that the size of the global client-side model is the same as that of the individual client-side model. In this context, $|w_c|$ is the function of $L_c$ since the number of model parameters increases with the number of layers. Similar to $f_{FLOPs}$, due to the difficulty of analytical modeling, we also adopted regression-based modeling. 
Consequently, as shown in Fig. \ref{fig:FLOPs_Params}, we conclude that $|w_c|$ is a convex function of $L_c$, modeled as $|w_c|=c\exp{(dL_c)}$, where $c$ and $d$ are coefficient values. Then, with given $p^t_{n}$ and $p^r_{n}$ of transmission and receiving power consumption of client $n$, the overall communication energy for interacting with main server $E_{com-m:n}$ and fed server $E_{com-f:n}$ can be calculated by
\begin{equation}\label{eqn:energy_com_m}
E_{com-m:n}= p^t_{n}t_{s:n}+p^r_{n}t_{g:n}, \quad \forall n \in \mathcal{N}.
\end{equation}
\begin{equation}\label{eqn:energy_com_f}
E_{com-f:n}= p^t_{n}t_{cm:n}+p^r_{n}t_{gcm:n}, \quad \forall n \in \mathcal{N}.
\end{equation}
\subsubsection{Overall energy consumption}
The energy consumption for $E$ epochs of client $n$ with $M$ mini-batches is given by $E_{n}^{epoch} = EM(E_{cmp,n} + E_{com-m,n})$, $n \in \mathcal{N}$. Consequently, the overall energy consumption for $T$ rounds, covering the entire duration until the SFL system completes its operations, is given by
\begin{equation} \label{eqn:energy_tot}
  \footnotesize  E_{tot,n} = T ( E_{n}^{epoch} + E_{com-f:n}), n \in \mathcal{N}.
  \vspace{-5pt} 
\end{equation}
Then, in order to demonstrate the influence of $L_c$ and $d_n$ on $E_{tot,n}$, (\ref{eqn:energy_tot}) can be simplified as 

\begin{equation} \label{eqn:energy_tot_v1}
  \footnotesize  E_{tot,n} = T ( (d_nf_{FLOPs}(L_c)A_n + B_n) + c\exp{(dL_c)}C_n), n \in \mathcal{N},
  \vspace{-1pt} 
\end{equation}
  \vspace{-1pt} 

where $A_n = E\frac{p^c_n }{f_n \kappa}$, $B_n=EME_{com-m,n}$, and $C_n=\frac{mp^t_{n}}{r_{cm:n}}+\frac{mp^r_{n}}{r_{gcm:n}}$. Similarly, (\ref{eqn:energy_tot_v1}) can be further simplified to highlight only the impact of $d_n$ on $E_{tot,n}$, which is given by

\begin{equation} \label{eqn:energy_tot_v2}
  \footnotesize  E_{tot,n} = d_n H_n + I_n, n \in \mathcal{N},
  \vspace{-1pt} 
\end{equation}
  \vspace{-1pt} 

where $H_n=Tf_{FLOPs}(L_c)A_n$ and $I_n=T(B_n+C_nc\exp{(dL_c)})$.
\subsubsection{Shared incentives and utility function design}
Adhering to the principle of proportional sharing, the SFL model owner allocates an amount of incentives $R$ to each client $n \in \mathcal{N}$ in proportion to their data contributions. Accordingly, the specific amount of incentives $R_n$ allocated to client $n$ by the SFL model owner is determined as
\begin{align} \label{eqn:R_n}
R_n=R\frac{d_n}{\sum_{l = 1}^N d_l}.
\end{align}
Thus, the utility function of client $n$ is given by
\begin{align} \label{eqn:utility_client}
U_n(\textbf{d}, R, L_c)&=\psi_n R_n- E_{tot,n} + S, \nonumber\\
                     &=\psi_n R\frac{d_n}{\sum_{l = 1}^N d_l}-d_n H_n - I_n + S, 
\end{align}
where $\psi_n (> 0)$ is a weighting factor for the incentives of client $n$. And $S$ is an offset to $U_n$ as a baseline constant term, representing the satisfaction of clients using the SFL model. 

Following the literature on designing incentive mechanisms in FL \cite{Lee22, Zhan2020}, this paper also adopts a linear function with respect to $R_n$ for clients, reflecting the fact that the incremental benefit from additional rewards remains constant, while the associated energy penalty increases linearly with the total energy expenditure $E_{tot,n}$.

 \begin{figure}       
	\centering	\includegraphics[width=220pt,keepaspectratio]{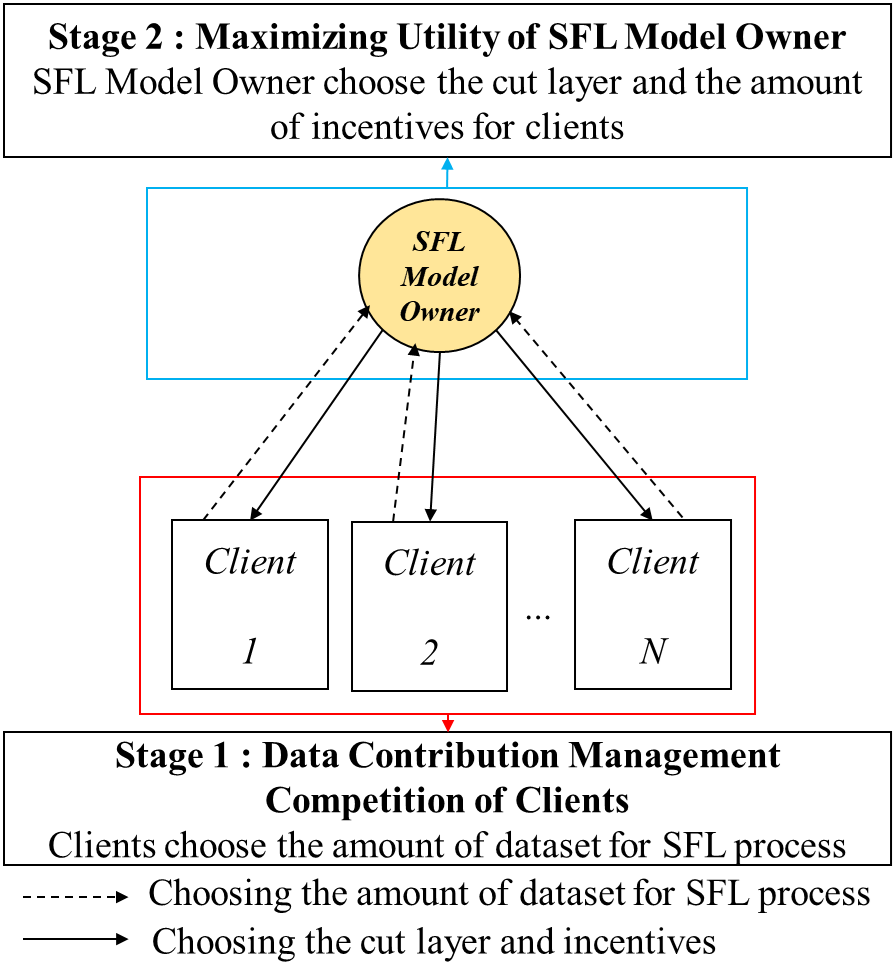}
    \caption{A diagram of a two-stage Stackelberg game in SFL where the SFL model owner is a leader and clients are followers.}
    \label{fig:Figure_diagram_game}
   \end{figure}

\section{Game-Theoretic Analysis}
In this section, we examine the competitive dynamics between the SFL model owner and the clients, while also considering the competition among clients to gain more incentives. These competitive behaviors can be considered as a hierarchical noncooperative decision problem, analyzable as a Stackelberg game. Stackelberg games are a type of noncooperative game characterized by a hierarchy among players, in which they are classified as leaders and followers \cite{Basar1999, Osborne1994, Roh2013a}. In this study, we consider a scenario where the clients follow the behavior of the SFL model owner. Treating the SFL model owner as the leader and the clients as the followers, the competition among players to maximize their utilities is modeled as a two-level Stackelberg game. If the SFL model owner decides on the cut layer and the total amount of incentives, the clients will respond by selecting their optimal strategies regarding data contribution for the SFL process. The SFL model owner is aware that the clients will choose their best response strategies based on the owner’s decisions. Therefore, the SFL model owner optimizes its strategy to maximize utility, taking into account the clients’ best responses. To find the Stackelberg Equilibrium (SE), which is the solution to this game, we employ a backward induction technique introduced in \cite{Basar1999, Osborne1994}. As illustrated in Fig. \ref{fig:Figure_diagram_game}, this technique can be summarized as follows. First, in the lower-level game for the clients, referred to as Stage 1: Data Contribution Management Competition, the clients determine the amount of data to contribute to maximize their utility function (\ref{eqn:utility_client}), given the incentives and cut layer configuration. Then, in the upper-level game for the SFL model owner, referred to as Stage 2: Maximizing Utility of SFL Model Owner, the SFL model owner combines the clients' strategies with its utility function (\ref{eqn:utility_MO}) and selects the cut layer and incentive amounts that maximize its own utility function (\ref{eqn:utility_MO}).

\subsection{Stage 1: Noncooperative Game Among Clients}
Each client aims to maximize its utility function (\ref{eqn:utility_client}) within the dataset constraints [0, $D_n$]. Given the vector $\mathbf{d} = (d_1, d_2, \dots, d_N)$, let $\mathbf{d}_{-n}$ denote the profile of strategies for the clients except $n$, i.e., $\mathbf{d}_{-n}=(d_l)_{l\in\mathcal{N}\backslash \{n\}}$. Then, by definition, we have $\mathbf{d}=(d_n, \mathbf{d}_{-n})$. 

Note that, to simplify the problem, we relax the integer variable $d_n$ into a continuous variable. Thus, the optimization problem to find the optimal amount of data contribution for the SFL process is given by
\begin{align}\label{eqn:COP} 
&\underset{d_n}{\textnormal{maximize}} \;U_n(\textbf{d}, R, L_c)\\ \nonumber
&\textnormal{ subject to }\\ \nonumber
& 0\le d_n \le D_n, \textnormal{ for } n \in \mathcal{N}. \\ \nonumber
\end{align}

\begin{definition} \label{Def:BestResponse}
The best response (BR) function $B_n(\mathbf{d}_{-n}, R, L_c)$ of client $n$ as a follower is the best strategy for client $n$ given the other clients' strategies $\mathbf{d}_{-n}$ and SFL model owner strategies $R$ and $L_c$. The BR function is denoted as follows:
\begin{eqnarray}\label{eqn:BestResponse}
B_n(\mathbf{d}_{-n}, R, L_c) =  \argmax_{d_n} U_n(d_n, \mathbf{d}_{-n}, R, L_c) \textnormal{ for } n \in \mathcal{N}.
\end{eqnarray}

\end{definition}

\begin{definition} \label{Def:NashEquilibrium}
A Nash equilibrium of the noncooperative
game among the clients is a profile of strategies $\mathbf{d}^*$ = $(d^*_{1},d^*_{2},\dots, d^*_{N})$
with the property that, given the SFL model owner strategies $R$, and $L_c$, we obtain
\begin{eqnarray}\label{eqn:BestResponse}
d^*_{n} = B_n(\mathbf{d}_{-n}, R, L_c),  \textnormal{ for } n \in \mathcal{N} .\\ \nonumber
\end{eqnarray}
\end{definition}

\begin{theorem}\label{eqn:concavityproof}
The utility function $U_n(d_n, \mathbf{d}_{-n}, R, L_c)$ of client $n$ is strictly concave in $d_n$. Thus, a Nash equilibrium exists in the noncooperative game among the clients.
\end{theorem}
\begin{proof} 
Taking the first and second derivatives of $U_n(d_n, \mathbf{d}_{-n}, R, L_c)$ with
respect to $d_n$, we obtain
\begin{align}\label{eqn:1stderivative}
    \frac{\partial U_n(d_n, \mathbf{d}_{-n}, R, L_c)}{\partial d_n} =  \psi_n R \frac{ \sum_{l \in \mathcal{N} \backslash \{n\}} d_{l} }{(\sum_{l = 1}^{N} d_{l})^2 } - H_n,
\end{align}

\begin{eqnarray}\label{eqn:2ndderivative}
    \frac{\partial^2 U_n(d_n, \mathbf{d}_{-n}, R, L_c)}{\partial d_{n}^2}
    =  - 2 \psi_n R  \frac{  \sum_{l \in \mathcal{N} \backslash \{n\}} d_{l} }{(\sum_{l = 1}^{N} d_{l})^3 } < 0,
\end{eqnarray}

It follows from \eqref{eqn:2ndderivative} that the second derivatives of the utility function $U_n$ is negative. This implies that the utility function $U_n(d_n, \mathbf{d}_{-n}, R, L_c)$ is strictly concave in $d_n$ where the dataset constraints [0, $D_n$] is a compact and convex set. Thus, the proposed noncooperative client-level game has a Nash equilibrium \cite{Lee2015}.
\end{proof}

By equating the first derivative of $U_n(d_n, \mathbf{d}_{-n}, R, L_c)$ in (\ref{eqn:1stderivative}) to zero, we obtain
$$
 \psi_n R \frac{ \sum_{l \in \mathcal{N} \backslash \{n\}} d_{l} }{(\sum_{l = 1}^{N} d_{l})^2 } = H_n.
$$

Accordingly, the BR function $B_n(\mathbf{d}_{-n}, R, L_c)$ (also denoted as $d_n^*$) of client $n$ is obtained as

\begin{eqnarray}\label{eqn:BestResponse2}
   d_n^{*} = \sqrt{\frac{ \psi_n R \sum_{l\in \mathcal{N}\backslash \{n\}} d_{l}}{H_n} } - \sum_{l\in \mathcal{N}\backslash \{n\}} d_{l}\nonumber\\
= \sqrt{\sum_{l\in \mathcal{N}\backslash \{n\}} d_{l}} \left ( \sqrt{\frac{\psi_n R}{H_n}} - \sqrt{\sum_{l\in \mathcal{N}\backslash \{n\}} d_{l}} \right ).
\end{eqnarray}

\begin{theorem}\label{Thm:unique_NE}
For the noncooperative game among the clients, the Nash equilibrium is characterized by the following closed-form expression:
\begin{align}\label{eqn:closedform_NE}
d_{n}^*=\frac{(N-1)R}{\sum_{l\in \mathcal{N}}\frac{H_l}{\psi_l}}\left(1-\frac{H_n(N-1)}{\psi_n\sum_{l\in \mathcal{N}}\frac{H_l}{\psi_l}}\right)\ \  \forall n \in \mathcal{N},
\end{align}
\end{theorem}
\begin{proof}
At the Nash equilibrium, the BR condition in
(\ref{eqn:BestResponse2}) is satisfied for every client. Thus, for any client $n\in \mathcal{N}$, we obtain
\begin{eqnarray}\label{eqn:sum_of_BR}
  \sum_{l\in \mathcal{N}}d_{l}^* = \sqrt{\frac{ \psi_n R \sum_{l\in \mathcal{N}\backslash \{n\}} d_{l}^*}{H_n} }.
\end{eqnarray}
By setting $\eta^*=\sum_{l\in \mathcal{N}}d_{l}^*$, we can derive that
\begin{eqnarray}\label{eqn:closedform_NE_v1}
  d_{n}^*=\eta^*-\frac{{\eta^*}^2H_n}{\psi_n R}.
\end{eqnarray}
By summing (\ref{eqn:closedform_NE_v1}) over $N$ clients, we have
\begin{eqnarray}\label{eqn:closedform_NE_v2}
  \eta^*=N\eta^*-\frac{{\eta^*}^2}{R}\sum_{l\in \mathcal{N}}\frac{H_l}{\psi_l}. 
\end{eqnarray}
By solving the problem of (\ref{eqn:closedform_NE_v2}), we obtain
\begin{eqnarray}\label{eqn:closedform_NE_v3}
 \eta^*=\frac{(N-1)R}{\sum_{l\in \mathcal{N}}\frac{H_l}{\psi_l}}. 
\end{eqnarray}
By substituting (\ref{eqn:closedform_NE_v3}) into (\ref{eqn:closedform_NE_v1}), we derive the closed-form expression in (\ref{eqn:closedform_NE}).
\end{proof}

The closed-form Nash equilibrium in (\ref{eqn:closedform_NE}) is projected onto the feasible region $0 \le d_n^* \le D_n$ whenever it falls outside the feasible region. Accordingly, the equilibrium contribution can be expressed as
\begin{eqnarray}\label{eqn:BR_sol}
     d_n^{*} =\left\{
                  \begin{array}{llll}
                   \displaystyle
                    0  &, & \hbox{\textnormal{if } $d_n^* < 0$ } & \\
                   (\ref{eqn:closedform_NE}) &, & \hbox{\textnormal{if } $0 \le d_n^* \le D_n$},& \\
                   D_n &, & \hbox{\textnormal{otherwise}}&
                   \end{array}
                \right.
\end{eqnarray}

\begin{theorem}\label{Thm:uniqueNE2}
The non-cooperative game among the clients has a unique Nash equilibrium.
\end{theorem}
\begin{proof}
Referring to Theorem \ref{Thm:unique_NE}, (\ref{eqn:closedform_NE_v3}) uniquely determines $\eta^*$, and (\ref{eqn:closedform_NE_v1}) then uniquely determines each $d_n^*$. Therefore, the closed-form of Nash equilibrium characterized in (\ref{eqn:closedform_NE}) is unique\footnote{The closed-form equilibrium in (\ref{eqn:closedform_NE}) uniquely determines the unconstrained equilibrium. The constrained equilibrium is then uniquely obtained by the projection rule in (\ref{eqn:BR_sol}).}.

\end{proof}

\subsection{Stage 2: Utility Maximization for SFL Model Owner}
The optimization problem for the SFL model owner is formulated as
\begin{align}\label{eqn:SFL_opt} 
&\underset{R, L_c}{\textnormal{maximize}} \;U_{MO}(R, L_c, \mathbf{d}^*)\\ \nonumber
&\textnormal{ subject to }\\ \nonumber
& R_{\min}\le R \le R_{\max} \ \& \ L_{\min}\le L_c \le L_{\max}, 
\end{align}
where $\mathbf{d}^*$ is given by  (\ref{eqn:closedform_NE}). Note that, to simplify the problem, we relax the integer variables $R$ and $L_c$ into continuous variables. In this context, $R_{\min}$ and $R_{\max}$ in the constraints can be set by the SFL model owner as the minimum requirement to motivate the SFL participants and as the maximum reward for attempting to achieve a positive utility value, respectively, by balancing the three terms in $U_{MO}$, taking their weighting factors into account. Moreover, $L_{\min}$ and $L_{\max}$ can be set by the SFL model owner to ensure that privacy concerns from clients (e.g., reconstruction attacks) are not violated and to avoid imposing excessive burdens on clients during the SFL process, respectively.

\begin{theorem}\label{eqn:global_optimal_SFL}
An optimal solution for $R$ and $L_c$ exists for maximizing the utility function $U_{MO}(R, L_c, \mathbf{d}^*)$.
\end{theorem}
\begin{proof} 
Since $d_n^*$ in (\ref{eqn:closedform_NE}) is a linear function of $R$, we can redefine it as
\begin{eqnarray}\label{eqn:simple_form_of_d*}
   d_n^* = X_n(L_c)R,
\end{eqnarray}
where $X_n$ is a function of $L_c$ that acts as the coefficient in the linear function specifying $d_n^*$. Then, by substituting (\ref{eqn:simple_form_of_d*}) into $U_{MO}$, then $U_{MO}$ is given by 
 \begin{align*}
    &U_{MO}(R, L_c, \mathbf{d}^*) \\
    &= \tau_1 \ln(1+R\frac{\sum_{n = 1}^N X_n(L_c)}{d_{req}}) +\tau_2 \frac{f_{FLOPs}(L_c)}{w_{FLOPs}} - R.
    \end{align*}
Finally, for any feasible $R$ and $L_c$, taking the second derivatives of $U_{MO}$ with respect to $R$ and $L_c$, we have
\begin{eqnarray}
\frac{\partial^2 U_{MO}}{\partial R^2}= -\frac{\tau_1 (\sum_{n=1}^N X_n(L_c))^2}{d_{req}^2(1 + R\frac{ \sum_{n=1}^N X_n(L_c)}{d_{req}})^2}<0.
\end{eqnarray}
Accordingly, $U_{MO}$ is strictly concave with respect to $R$ with given $X_n(L_c)$. Then, the optimal $R^*$ can be determined using a CVX solver. In this scalar decision variable case, the computational complexity of the optimization scales logarithmically with the desired solution accuracy $\epsilon$, i.e., $O(\log(1/\epsilon))$ \cite{Boyd04}. On the other hand, demonstrating the concavity of the function $X_n(L_c)$ with respect to $L_c$ analytically is challenging. Therefore, within the constraints from $L_{\min}$ to $L_{\max}$, we can find the optimal solution of $R^*$ and $L_c^*$ using the exhaustive search-based heuristic algorithm with the CVX solver described in Algorithm \ref{algo-Ex}. 

Then, the overall computational complexity of the server-side optimization becomes $(1+L_{\max}-L_{\min})O(\log(1/\epsilon))$. Since the clients' equilibrium responses are available in closed form, no iterative best-response computation is required. Evaluating the closed-form responses for all $N$ clients requires only $\mathcal{O}(N)$ arithmetic operations.
The corresponding end-to-end computational scalability is empirically evaluated in Section~\ref{sec:experiments}.

As a result, given that the search space of $L_c$ is relatively small, the computational latency incurred by solving the server-side optimization problem is negligible compared to the communication overhead and model training time in the overall framework. Consequently, the proposed optimization procedure does not introduce a performance bottleneck, ensuring that the system remains practically deployable.

\begin{algorithm}[t]
\caption{Find optimal $L_c$ and $R$}
\label{algo-Ex}
\hspace*{\algorithmicindent} \textbf{Input :} $\mathbf{d}^*$ \\
\hspace*{\algorithmicindent} \textbf{Output :} Optimal $L_c$, $R$
\begin{algorithmic}[1]
    \State \textbf{Initialize} \textbf{M}$[1, \dots, L_{\max}]$ with $-Inf$
    \For {$L_c$ from $L_{\min}$ to $L_{\max}$}
		\State $R \gets $ solve (\ref{eqn:SFL_opt}) by using a CVX solver
            \State \textbf{M}$[L_c]$ $\gets U_{MO}(R, L_c, \mathbf{d}^*)$
	\EndFor	
 
	\State $(R^*, L_c^*) \gets$ argmax($\textbf{M}$)
	\State Return $(R^*, L_c^*)$
\end{algorithmic}
\end{algorithm}

\end{proof}

\begin{theorem}
A Stackelberg equilibrium exists in the proposed two-stage Stackelberg game in SFL.
\begin{proof}
Because the noncooperative game among the clients has a unique Nash equilibrium $\mathbf{d}^* = (d_1^*, d_2^*, \dots, d_N^*)$, as shown in Theorem \ref{Thm:uniqueNE2}, and the SFL model owner can always find its optimal strategy $R^*$ and $L_c^*$, as shown in Theorem \ref{eqn:global_optimal_SFL}, we conclude that a Stackelberg equilibrium exists in the proposed two-stage Stackelberg game in SFL.
\end{proof}
\end{theorem}

\section{Performance Evaluation}
\label{sec:experiments}

\begin{table}[]
\caption{Parameter Settings} 
\centering 
\begin{tabular}{c c}
\hline\hline 
Parameters & Values
\\ [0.5ex] 
\hline 
$N$ & 5\\ 
$M$ & 100\\ 
$d_{req}$ & 2,000 \\
$m$ & 32 \\
$s_n, g_n$ & 24.5 Mbits\\
$r_{s:n}, r_{cm:n}$ & 100 Mbit/s\\
$r_{g:n}, r_{gcm:n}$ & 200 Mbit/s\\
$w_{FLOPs}$ & 4.3 GFLOPs\\
$\kappa$ & 16 FLOPs/cycle\\
$p_n^c$ & 4 W\\
$p_n^t, p_n^r$ & 0.2 W\\
\hline 
\end{tabular}
\label{table:parameter} 
\end{table}

\begin{table}[]    
    \caption{Regression-Based Models}
    \centering
    \begin{tabular}{l|c|c}
    \hline\hline  
      & Proposed Models & RMSE 
    \\ [0.5ex]
    \hline 
    $GFLOPs$  & $f_{\text{FLOPs}}(L_c) = 0.3779 \cdot L_c - 0.212$ & $0.0004 $ \\
    $MParams$ & $|w_c| = 0.1098 \cdot \exp(0.4711 \cdot L_c)$ & $0.0909 $ \\
    \hline 
    \end{tabular}
    \label{table:regression} 
\end{table}

\begin{figure}       
\centering	\includegraphics[width=220pt,keepaspectratio]{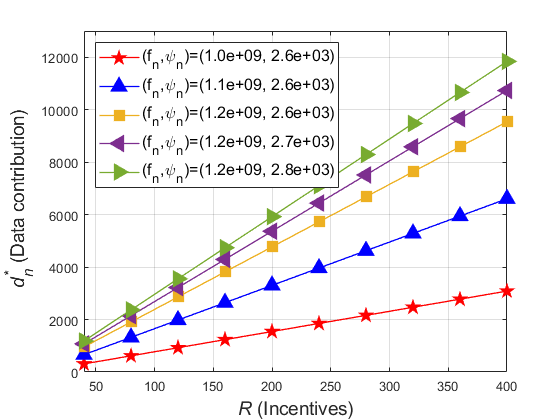}
\caption{Data contribution at NE for various ($f_n$,  $\psi_n$ ) pairs.}
\label{fig:graph_1}
\end{figure}

\begin{figure}[]
    \centering
    \begin{subfigure}{0.47\textwidth}
		\centering
		\includegraphics[width=\textwidth]{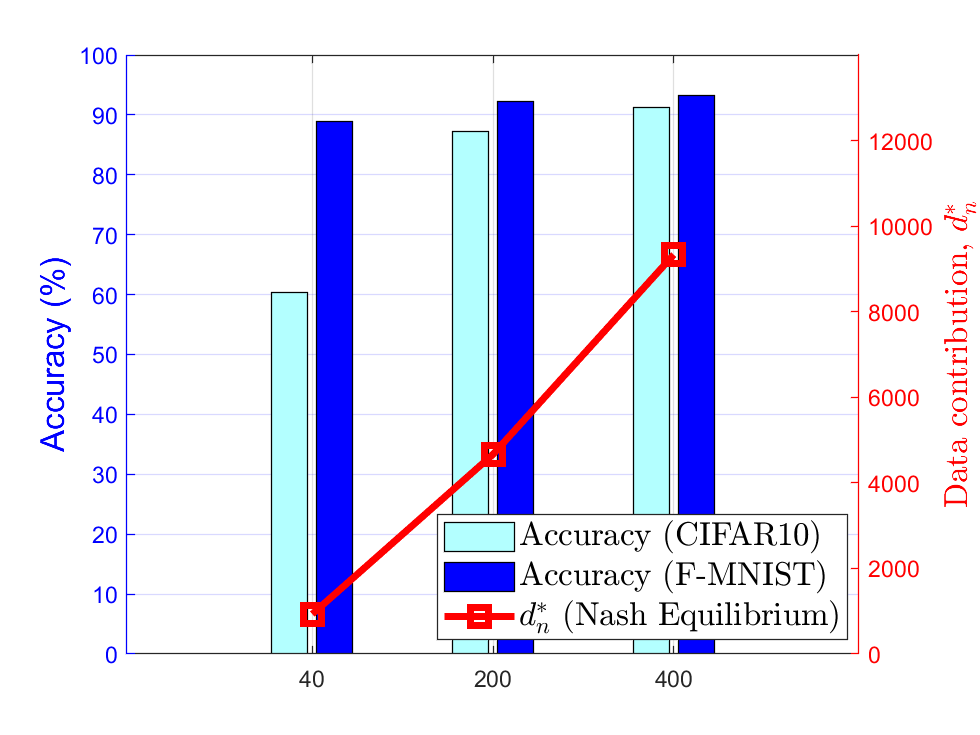}
		\caption{F-MNIST/CIFAR10}
		\label{fig:graph_2-1}
	\end{subfigure}
    \centering
    \begin{subfigure}{0.47\textwidth}
		\centering
		\includegraphics[width=\textwidth]{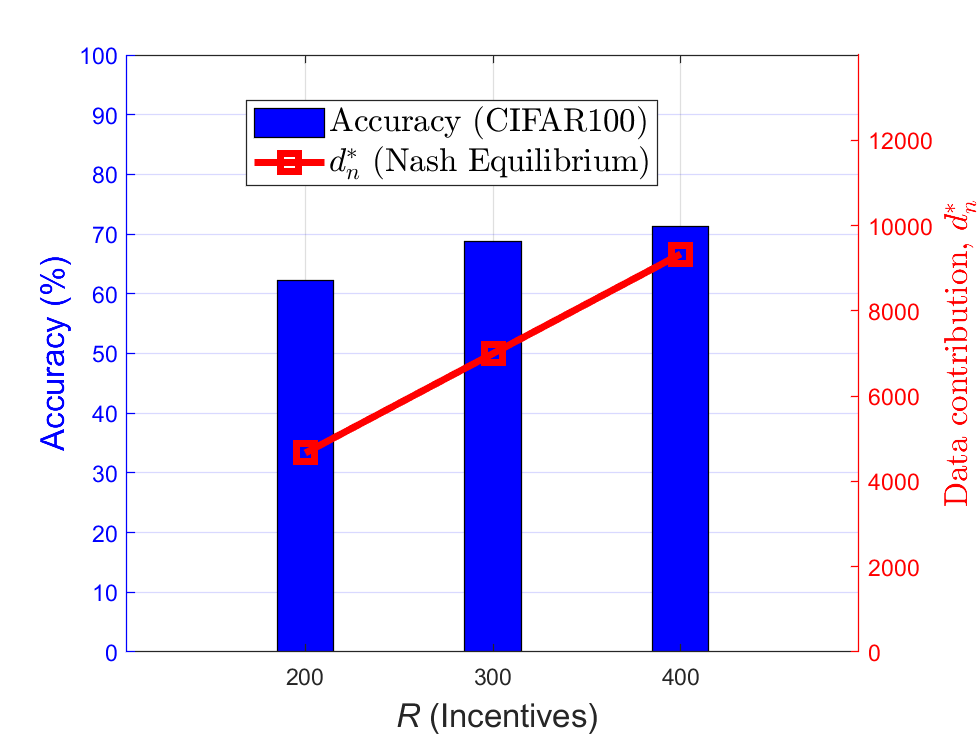}
		\caption{CIFAR100}
		\label{fig:graph_2-2}
	\end{subfigure}
    \begin{subfigure}{0.47\textwidth}
		\centering
		\includegraphics[width=\textwidth]{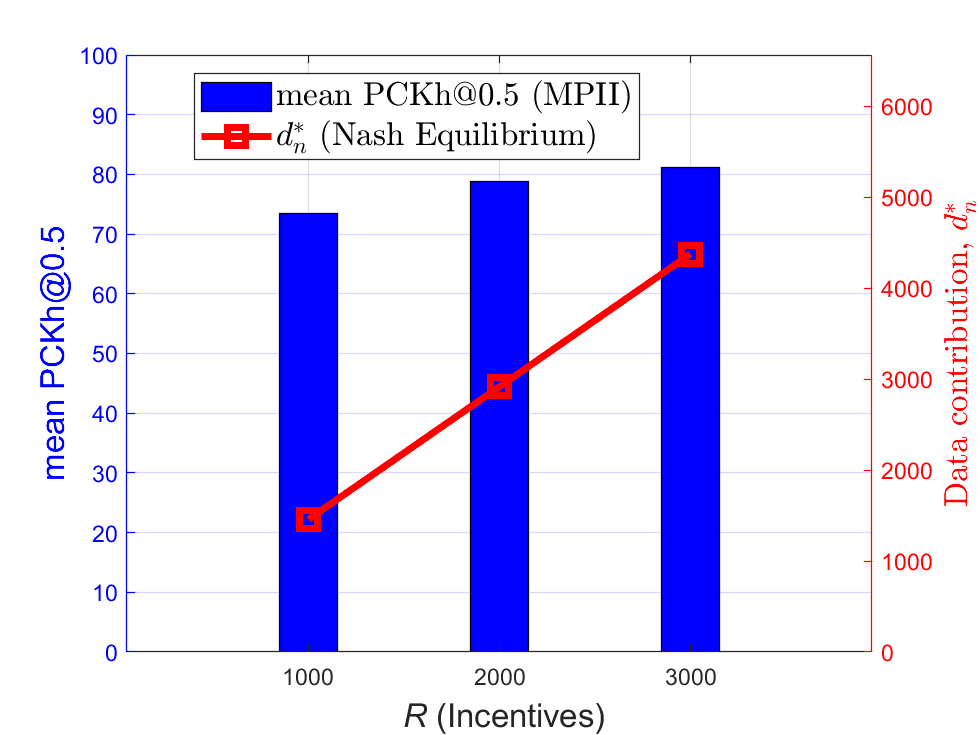}
		\caption{MPII (Human Pose Estimation)}
		\label{fig:graph_2-3}
	\end{subfigure}
    \caption{Equilibrium data contribution and resulting model accuracy.}
    \label{fig:graph_2}
\end{figure}

This section presents experimental results from our game-theoretic analysis of SFL, involving the SFL model owner and the clients. The detailed simulation parameters and regression-based models are specified in Tabs. \ref{table:parameter} and \ref{table:regression}. The values are chosen following settings in prior work~\cite{lin2024efficient,li2023energy,keshari2018learning} and our implementation details (FP32 default in PyTorch~\cite{paszke2017automatic}, FLOPs profiled by THOP~\cite{thop}).

We begin by examining the optimal data contribution $d_n^*$ from the $n$-th client at the Nash Equilibrium (NE). For given incentive values $R$, the client data contributions are derived from (\ref{eqn:closedform_NE}) and (\ref{eqn:BR_sol}).

Thus, Figs.~\ref{fig:graph_1} and~\ref{fig:graph_2} report the
equilibrium data contributions $d_n^*$ associated with each incentive
$R$.

For this numerical analysis, we assume that the cut layer is set at $L_c=3$, and five heterogeneous clients have CPU capabilities $f_n = [1.0, 1.1, 1.2, 1.2, 1.2]$ GHz, with corresponding incentive weighting factors $\psi_n = [2.6, 2.6, 2.6, 2.7, 2.8]\times 10^3$, respectively. Fig. \ref{fig:graph_1} shows the results for $R$ in the range $[40, 400]$. As expected, incentive $R$ increases the data contribution $d_n^*$ by encouraging clients to participate more actively in SFL process.
Interestingly, $d_n^*$ diminishes as the CPU capability $f_n$ of participating clients increases, suggesting a strategic balance between data contribution and energy consumption, as modeled in (\ref{eqn:utility_client}). This is because clients with higher CPU capabilities tend to reduce execution times more, balancing the trade-off between earning incentives and minimizing energy consumption. Notably, among clients with identical CPU capabilities, those with higher incentive weighting factors $\psi_n$ tend to contribute more data.

\begin{figure*}[t]
    \centering
    \begin{subfigure}{0.32\textwidth}
		\centering
		\includegraphics[width=\textwidth]{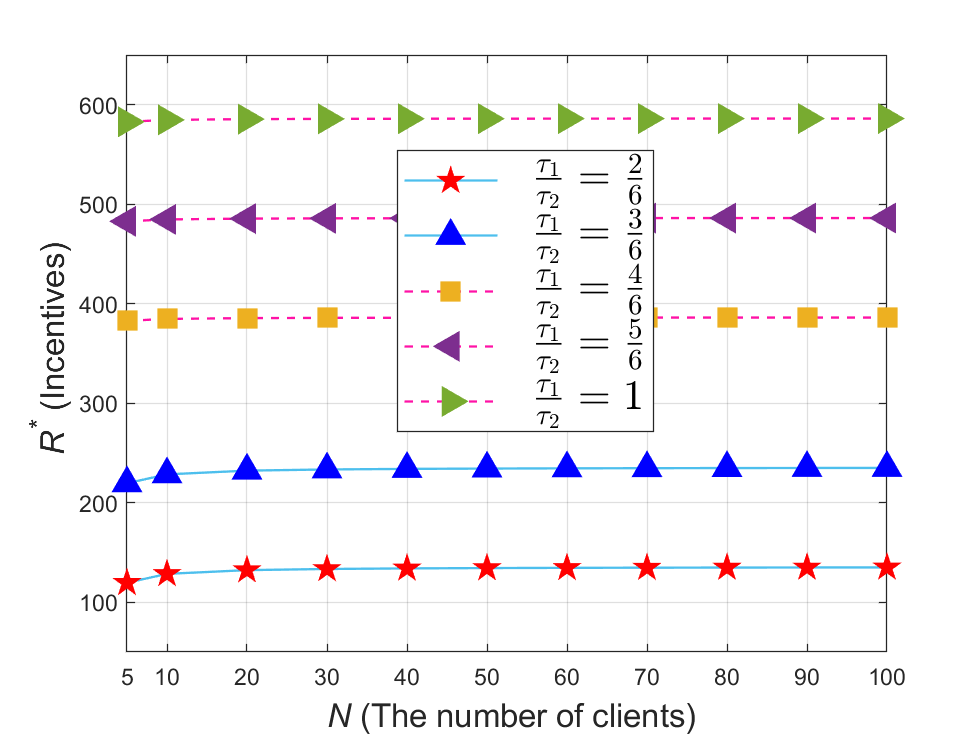}
		\caption{Incentives $R$}
		\label{fig:graph_3}
	\end{subfigure}
    \centering
    \begin{subfigure}{0.32\textwidth}
		\centering
		\includegraphics[width=\textwidth]{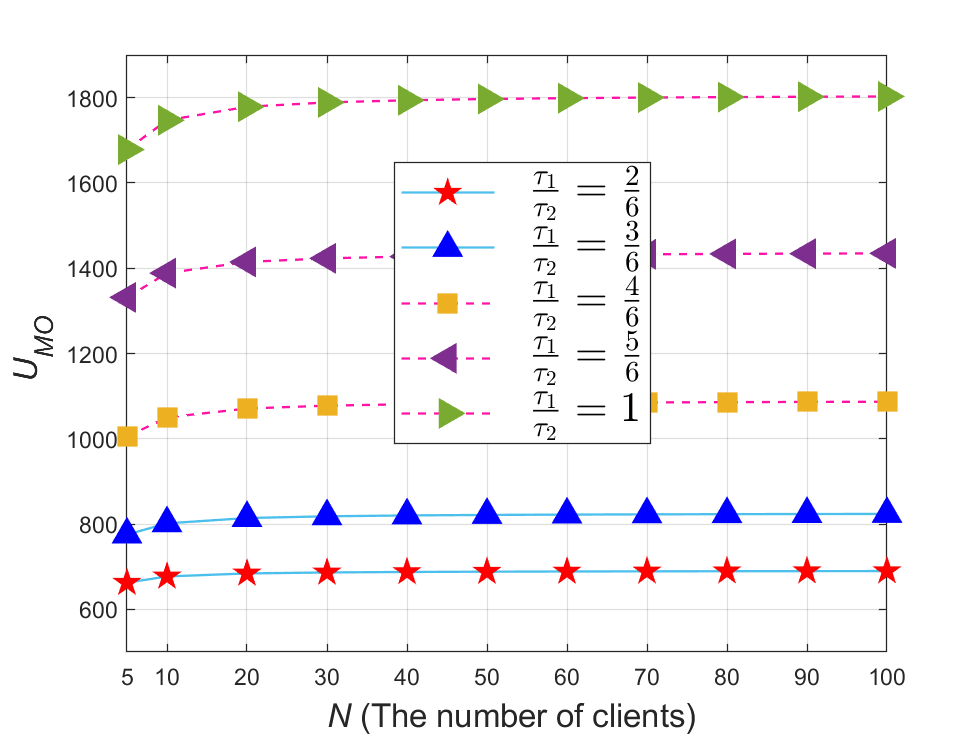}
		\caption{Utility function $U_{MO}$}
		\label{fig:graph_4-1}
	\end{subfigure}
    \begin{subfigure}{0.32\textwidth}
		\centering
		\includegraphics[width=\textwidth]{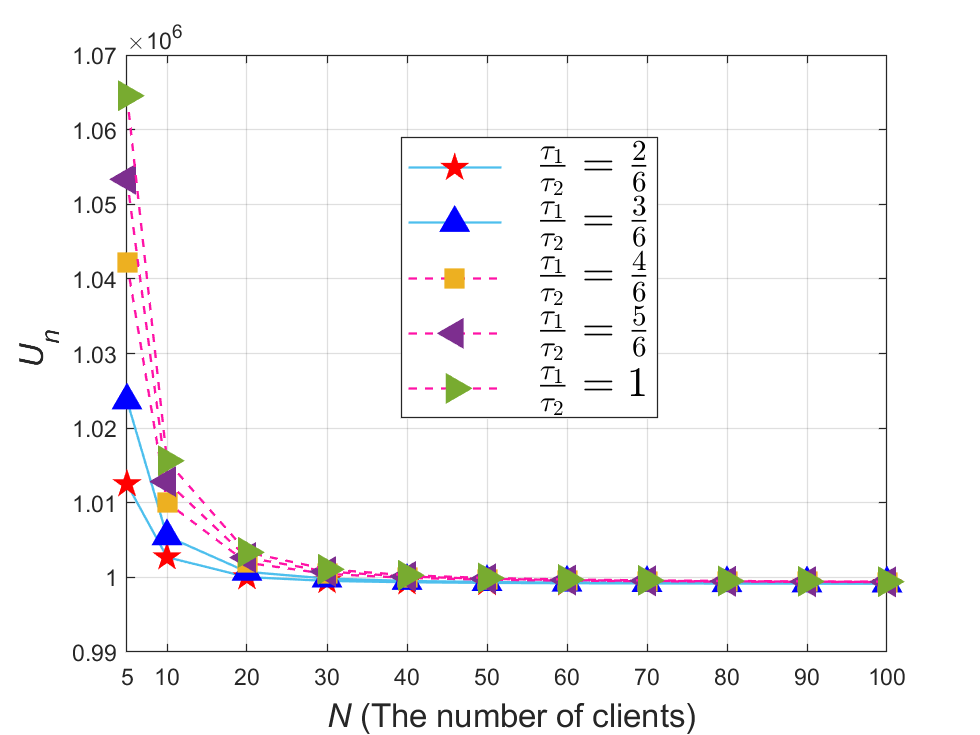}
		\caption{Utility functions $U_n$}
		\label{fig:graph_4-2}
	\end{subfigure}
    \caption{(a) Incentives at NE and (b)-(c) Utility functions at SE for various $\frac{\tau_1}{\tau_2}$ ratios. (For the dotted line, the cut layer $L_c^*$ was selected as $L_{\min}$, whereas for the solid line, it was chosen as $L_{\max}$.)}
    \label{fig:graph_4}
\end{figure*}

To keep the experiment manageable with classification accuracy based on data contributions, we further assume homogeneity among the five clients, with each having $f_n = 1.2$ GHz and $\psi_n = 2.8 \times 10^3$. For this analysis, we conducted classification experiments under Independent and Identically Distributed (IID) data settings using two datasets, the CIFAR-10 dataset \cite{krizhevsky2009learning} and Fashion-MNIST \cite{xiao2017/online} to evaluate SFL accuracy according to the data contribution. The CIFAR-10 dataset comprises color images with 10 object categories, with 50,000 and 10,000 samples in training and test splits. 
The Fashion-MNIST dataset comprises grayscale images categorized into 10 distinct fashion classes, with 60,000 samples for the training split and 10,000 samples for the test set. To ensure consistency in the number of training samples between Fashion-MNIST and CIFAR-10 datasets, we utilized 50,000 samples for the training set from the Fashion-MNIST training split. This adjustment facilitates a direct comparison of accuracy trends between the two datasets, enabling a consistent evaluation of classification accuracy in relation to data contributions.
Data from the training split were allocated among five clients proportionate to their data contributions $d_n^*$, which were determined by incentive values $R$. Specifically, in scenarios with incentives of 40, 200, and 400, each client contributes 900, 4700, and 9300 samples, respectively, totaling 4,500, 23500, and 46,500 samples were used for SFL training in each scenario. 

The model evaluation employed all 10,000 samples from the test split. The SFL model $w$ is structured with 12 convolutional blocks (Conv-BN-ReLU), incorporating Max Pooling every third block, followed by two fully connected layers, a dropout layer, and a Softmax output layer. The model split between client and server was determined at cut layers after the convolution blocks, offering 12 viable points for split. Input images were normalized using predefined means and standard deviations (mean=[0.4914, 0.4822, 0.4465], std=[0.2023, 0.1994, 0.2010]), and data augmentation techniques such as random cropping and horizontal flipping were applied. An Adam optimizer facilitated the training, set to a learning rate of 0.0001 and beta values of (0.5, 0.999). Model training hyper-parameters are detailed in Tab \ref{table:parameter}. Fig.~\ref{fig:graph_2-1} jointly presents the equilibrium data contribution $d_n^*$ and the resulting model accuracy for different incentive values $R$. As $R$ increases, the proposed incentive mechanism induces a larger equilibrium contribution $d_n^*$, and the SFL model trained with the corresponding contributed data achieves higher accuracy. While the positive effect of a larger training set on model accuracy is consistent with prior studies~\cite{Niyato_2016,Luo2020,Kim2022,Doyeon2023}, the purpose of Fig.~\ref{fig:graph_2-1} is to explicitly demonstrate the linkage from the incentive $R$ to the equilibrium data contribution $d_n^*$ and subsequently to the resulting learning performance.

To enhance the generalizability of these findings, we extended our experiments to the CIFAR-100 dataset \cite{krizhevsky2009learning}, which is similar to CIFAR-10 but consists of 100 object categories. The dataset is divided into 50,000 training samples and 10,000 test samples. Since datasets with larger numbers of categories generally require more training data to achieve stable model training, we set the incentives $R$  to 200, 300, and 400. This resulted in each client contributing 4,700, 7,000, and 9,300 samples, respectively. The SFL model in the CIFAR-100 scenario was trained using a total of 23,500, 35,000, and 45,000 samples. In addition, we conducted experiments on human pose estimation using the MPII dataset \cite{andriluka20142d} to evaluate our method in a real-world SFL framework, similar to~\cite{lim2025mitigating}. Human pose estimation was formulated as a high-resolution regression task and therefore adopted an encoder-decoder architecture. We used the same 12 convolutional blocks as the encoder used in the classification experiments, while the decoder followed the architecture in~\cite{xiao2018simple} for joint-wise heatmap generation. Since this task requires larger input images, we used a resolution of $128\times128$, which is four times larger than CIFAR-10, leading to higher computational cost and thus requiring larger incentives. We set the incentives $R$ to 1,000, 2,000, and 3,000, resulting in each client contributing 1,500, 2,900, and 4,400 samples, respectively, for a total of 7,500, 14,500, and 22,000 training samples. Figs.~\ref{fig:graph_2-2} and~\ref{fig:graph_2-3} further show that the relationship among the incentive $R$, equilibrium data contribution $d_n^*$, and resulting model performance is consistently observed across different datasets, extending beyond classification to structured prediction tasks.

Continuing, we examine the numerical scalability of the proposed game-theoretic mechanism by increasing the number of clients from $N=5$ to $N=100$ under a homogeneous-client setting, with $f_n=1.2$ GHz and $\psi_n=2.8\times10^3$ for all $n$. The optimization constraints in (\ref{eqn:SFL_opt}) $L_{\min}$, $L_{\max}$, $R_{\min}$, and $R_{\max}$ are set to 3, 12, 60, and 1,000, respectively. Fig.~\ref{fig:graph_4} shows the numerical outcomes for various settings of the $\frac{\tau_1}{\tau_2}$ ratio over a substantially larger client population. 
Fig.~\ref{fig:graph_3} shows that the optimal incentive $R^*$ generally increases with the number of clients and gradually approaches a stable level as $N$ becomes large, since the incentive $R$ is proportionally shared among the participating clients according to their data contributions, as described in (\ref{eqn:utility_client}). Additionally, an elevated $\tau_2$ compared to $\tau_1$ (i.e., lower ratio $\tau_1$ to $\tau_2$) reflects the SFL model owner's strategy to minimize the server's computational burden, resulting in the maximal cut layer value $L_c=12$ and a correspondingly lower $R^*$.
In contrast, an increased $\tau_1$, enhancing data contribution, leads to an increase in the optimal $R^*$ and maintaining the minimal cut layer $L_c=3$ to encourage active client participation. These observations are further corroborated by additional analysis of the utility functions of the SFL model owner $U_{MO}$ and clients $U_n$.
As shown in Fig. \ref{fig:graph_4-1}, the utility trends of $U_{MO}$ for the SFL model owner are consistent with Fig. \ref{fig:graph_3}, whereas the pattern in Fig. \ref{fig:graph_4-2} contrasts with these\footnote{When plotting this graph, we set the baseline constant term $S$ to $10^6$ to ensure $U_n$ is greater than zero.}. It becomes evident that as more clients engage in SFL, their utility $U_n$ diminishes. Specifically, a higher ratio of $\tau_1$ to $\tau_2$ distinctly leads to reduced utility for the number of clients, confirming the inverse relationship highlighted in Fig. \ref{fig:graph_4-2}. Overall, these trends remain consistent across the evaluated range from $N=5$ to $N=100$, providing numerical evidence of the calability of the proposed Stackelberg framework.

\begin{figure}       
\centering	\includegraphics[width=220pt,keepaspectratio]{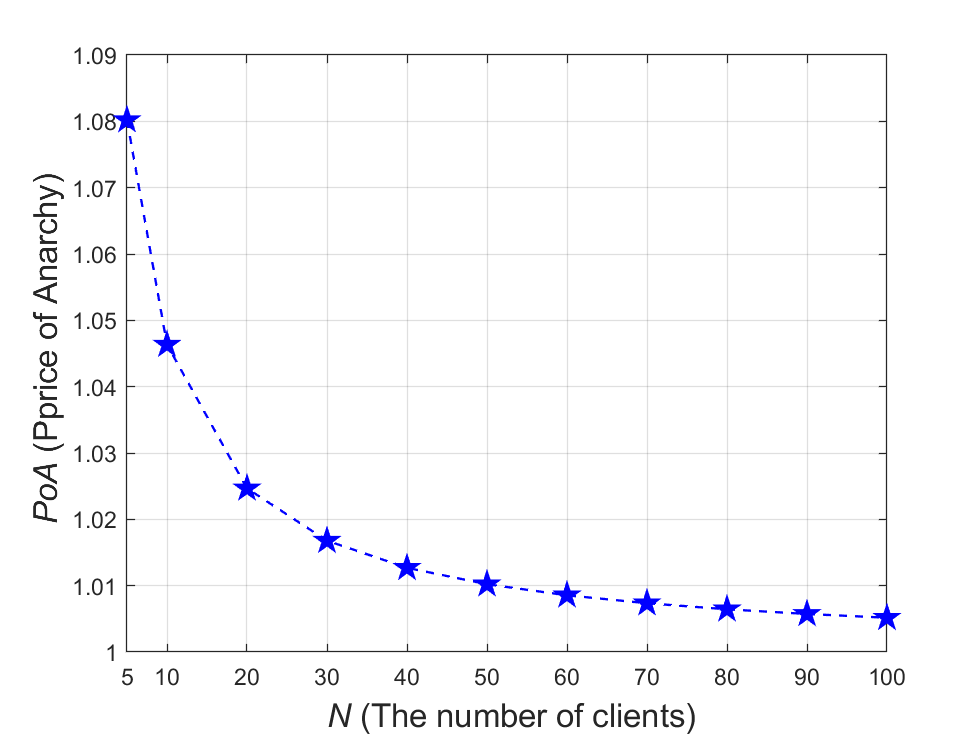}
\caption{PoA vs. the number of clients ($\psi_n$=2800)}
\label{fig:graph_5}
\end{figure}

\begin{figure}       
\centering	\includegraphics[width=220pt,keepaspectratio]{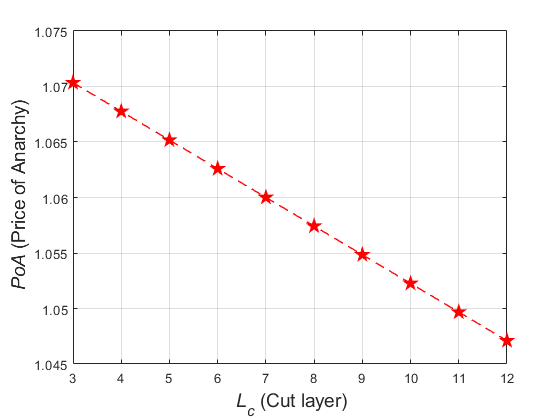}
\caption{PoA vs. varying $L_c$ $(N=6)$}
\label{fig:graph_6}
\end{figure}

\begin{figure}       
\centering	\includegraphics[width=220pt,keepaspectratio]{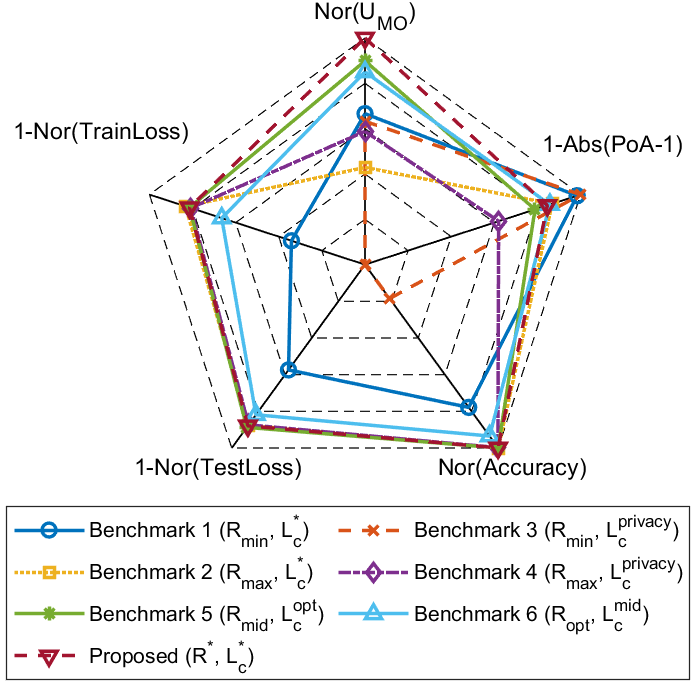}
\caption{Comparison of normalized utility, PoA, and performance metrics. Normalization is defined as $\text{Nor}(a)=a/max(a)$ and $\text{Abs}(a)=|a|$ represents the absolute value. The larger the pentagon, the better.}
\label{fig:graph_7}
\end{figure}

Finally, we evaluated the efficiency of a noncooperative game among clients, which is degraded due to their selfish behavior in the absence of central coordination. This efficiency is measured by the price of anarchy (PoA) [1, $\infty$), defined as the ratio between the optimal centralized solution and the worst equilibrium in games with multiple equilibria. In our case, since our problem has a unique NE, the PoA is given by
\begin{align} \label{eqn:PoA}
PoA&=\frac{\underset{d_n}{\textnormal{maximize}}\sum_{n=1}^NU_n(\textbf{d}, R, L_c)}{\sum_{n=1}^NU_n(\textbf{d}^*, R, L_c)}.
\end{align}

As illustrated in Fig.~\ref{fig:graph_5}, the PoA consistently decreases and approaches one as the number of clients increases from $N=5$ to $N=100$. 

This result shows that the efficiency of the NE in the noncooperative client-level game is maintained and further improves as the client population becomes larger\footnote{We set the minimum required $d_n$ as $d_{req}/N$ to ensure that the aggregated data contribution meets the minimum necessary amount for acceptable accuracy in SFL.}.
This occurs because, as the number of clients increases, the optimal data contribution $\mathbf{d}^*$ at NE decreases. The elevated competition among a larger number of clients reduces expected individual incentives, thereby narrowing the gap between the optimal centralized solution and the NE. To complement the above client-scale equilibrium analysis, we also
evaluated the computational scalability of obtaining the SE. Specifically, we measured the end-to-end computation time required to evaluate the closed-form client responses and optimize $R$ and $L_c$. The computation time remained approximately 5 s as the number of clients increased from $N=10$ to $N=1{,}000$, indicating that the additional overhead caused by
increasing the client population is negligible over the considered range.

Additionally, as shown in Fig. \ref{fig:graph_6}, the PoA also decreases with an increase in $L_c$. This occurs because a higher $L_c$ places a greater burden on clients, leading to a reduction of $\mathbf{d}^*$ at NE. Consequently, $L_c$ can be an effective regulator for controlling the effectiveness of the NE by adjusting the level of competition. To assess the effectiveness of the proposed joint optimization, we compared it with a set of controlled benchmark configurations, as illustrated in Fig.~\ref{fig:graph_7}.
Under the conditions $N=5$ and $\frac{\tau_1}{\tau_2}=\frac{4}{6}$, the proposed method determined the optimal incentive and cut layer to be $R^*=382.85$ and $L_c^*=3$, respectively. This configuration resulted in each client contributing 9,000 training samples, with a total of 45,000 samples across all five clients in the CIFAR-10 dataset.
The impact of the incentive $R$ was analyzed through two benchmarks: Benchmark 1, defined as $(R_{\min}, L_c^*)$, and Benchmark 2, defined as $(R_{\max}, L_c^*)$, where $L_c^*$ is the same cut layer obtained by the proposed method. With $R_{\min}=60$, each client contributed $d_n^* = 1,400$ data samples, while $R_{\max}$ yielded $d_n^* = 23,330$ samples per client. However, due to the CIFAR-10 dataset's constraints, the maximum number of samples per client was limited to 10,000.
The effect of the cut layer was further evaluated using Benchmark 3 and Benchmark 4, defined as $(R_{\min}, L_c^{privacy})$ and $(R_{\max}, L_c^{privacy})$, respectively, where $L_c^{privacy}$ is configurable to enforce a target privacy threshold (i.e., the Structural Similarity Index Measure (SSIM) used in this paper) in accordance with the SFL model owner’s privacy policy. Specifically, SSIM is a widely adopted metric for assessing reconstruction quality by comparing two images based on three components—luminance, contrast, and structural similarity. It can also be widely used to quantify the risk of privacy leakage under reconstruction attacks, in which an honest-but-curious (HBC) main server may attempt to recover the original input data from the smashed data (i.e., the embedded features) received from clients\footnote{$SSIM(x, y) = l(x,y)c(x,y)s(x,y)$, where $l(\cdot)$, $c(\cdot)$, and $s(\cdot)$ denote the luminance, contrast, and structural comparison functions, respectively. The SSIM value ranges from 0 to 1, with higher values indicating greater similarity between images. Consequently, an SSIM value close to 1 implies severe privacy leakage, whereas lower values indicate stronger privacy protection.} \cite{Lee2024a}. In these cases, $L_c^{privacy}=8$ was selected following the benchmark privacy setting in \cite{Lee2024b}, where the cut layer is determined based on an SSIM threshold of 0.6. To further examine the benefit of jointly optimizing the incentive and cut layer, we additionally consider two partial-optimization benchmarks. Benchmark~5 is defined as $(R_{\mathrm{mid}},L_c^{\mathrm{opt}})$, where $R_{\mathrm{mid}}=(R_{\min}+R_{\max})/2$ and $L_c^{\mathrm{opt}}$ is optimized while fixing $R=R_{\mathrm{mid}}$. Benchmark~6 is defined as $(R_{\mathrm{opt}},L_c^{\mathrm{mid}})$, where $L_c^{\mathrm{mid}}=\lfloor(L_{\min}+L_{\max})/2\rfloor=7$ and $R_{\mathrm{opt}}$ is optimized while fixing $L_c=L_c^{\mathrm{mid}}$. These two benchmarks provide controlled partial-optimization baselines in which only one decision variable is optimized while the other is fixed.

The proposed method and Benchmarks 1 to 6 were compared across five measures: utility of the SFL model owner ($U_{MO}$), price of anarchy (PoA), Accuracy, TestLoss, and TrainLoss. Each measure was normalized using $\text{Nor}(a) = \frac{a}{\max(a)}$, with $\text{Abs}(a) = |a|$ representing the absolute value. Fig. \ref{fig:graph_7} presents the normalized comparison of utility, PoA, accuracy, and loss values, where a larger pentagon plot signifies better overall performance.
As demonstrated by the previous experiments, larger training datasets lead to improve accuracy, indicating $R_{\max}$ maximizes accuracy (Benchmarks 2 and 4). However, when $R_{\max}$ is applied, this also results in a side effect of reduced $U_{MO}$. Furthermore, $R_{\max}$ also significantly increases data requirements and training complexity.
An increased $L_c$ causes clients to contribute fewer data samples, which in turn degrades classification accuracy. (Benchmark 3)
In contrast, the proposed method successfully balanced $R^*$ and $L_c^*$, optimizing both utility and accuracy. The optimal $R^*=382.85$ determined by the proposed method requires only $d_n^* = 9,000$ samples per client, achieving comparable performance to $R_{\max}$ case while offering improved utility ($U_{MO}$) for the SFL model owner.

Finally, efficiency was assessed using the PoA, where values closer to 1 indicate greater efficiency. As shown in Fig. \ref{fig:graph_7}, Benchmarks 1 and 3 achieve the highest PoA, but this comes with the sacrifice of reduced classification accuracy due to $R_{\min}$. The proposed method achieved a PoA closer to 1 than Benchmark 2 and Benchmark 4 ($R_{\max}$). The comparison with Benchmarks~5 and 6 further shows that jointly optimizing $R$ and $L_c$ provides a better balance among the considered metrics than optimizing either decision variable alone. We can conclude that the proposed method effectively balances $R$ and $L_c$, achieving both optimal utility and accuracy.

\begin{table}[t]
\centering
\caption{Ablation of the burden-aware client utility.}
\label{tab:burden_ablation}

\begin{tabular}{c|c|c|c|c}
\hline
Setting & Burden term & $d_n^*$ & $R^*$ & $L_c^*$ \\ \hline
Proposed & Included & $9{,}000$ & $382.85$ & $3$ \\
Burden-unaware & Removed & $D_n~(10{,}000)$ & $60$ & $12$ \\ \hline
\end{tabular}
\end{table}

To further examine the necessity of the burden-aware client utility, we consider an ablation case in which the computation and networking burden penalty is removed from the client utility in (12). Without this penalty, the client utility monotonically increases with the data contribution $d_n$, and therefore each client selects its maximum available data contribution, i.e., $d_n^*=D_n$. Consequently, the equilibrium data contribution becomes independent of both the incentive $R$ and the cut layer $L_c$. By substituting $d_n^*=D_n$ into the SFL model owner's optimization problem in (25), increasing $R$ no longer induces any additional data contribution and only increases the incentive expenditure. Therefore, the optimal incentive becomes $R^*=R_{\min}$. Moreover, since $f_{\mathrm{FLOPs}}(L_c)$ increases with $L_c$, the optimal cut layer becomes $L_c^*=L_{\max}$ under the considered setting, as summarized in Table~\ref{tab:burden_ablation}.

The burden-unaware case therefore collapses to a degenerate operating point in which every client contributes its entire available dataset regardless of the incentive and cut-layer decisions. This behavior ignores the computation and networking costs incurred by self-interested clients. The ablation therefore demonstrates the importance of the burden-aware utility in capturing the strategic trade-off between incentive rewards and participation costs.

\section{Discussion}

\subsection{Impact of Cut Layer Selection on Privacy}
In SFL, the main server can still potentially reconstruct the original input data from the smashed data (specifically, the embedded features) received from clients, a vulnerability known as a reconstruction attack. This risk is particularly relevant if the main server is HBC. From the study \cite{Lee2024a}, we found that the selection of the cut layer significantly impacts the privacy level. For example, as the cut layer $L_c$ increases, indicating greater model complexity, the quality of the reconstructed image decreases. This degradation in image quality corresponds to enhanced privacy because the increased complexity of the client-side model introduces advanced nonlinearities into the output. Therefore, to address client concerns regarding privacy, the SFL model owner should carefully select $L_{\min}$ to ensure an acceptable level of privacy. Accordingly, the proposed framework incorporates privacy through a predefined minimum cut-layer constraint rather than explicitly modeling it in the utility functions. Nevertheless, SSIM used in the evaluation measures image similarity rather than information-theoretic privacy leakage and therefore may not generalize across different learning tasks or threat models. Extending the proposed framework to directly integrate information-theoretic privacy metrics into the utility formulation, while accounting for cumulative privacy leakage over repeated SFL rounds, constitutes an important direction for future research.

\subsection{Interplay between Cut Layer Selection and Differential Privacy}
Clients can enhance their privacy by proactively adding noise to the smashed data, a technique known as differential privacy (DP). Here, the amount of noise added is inversely proportional to the privacy leakage parameter $\epsilon$. Formally, let $D$ represent a set of data points, and $\mathcal{M}$ be a probabilistic function or mechanism acting on databases. The mechanism $\mathcal{M}$ is considered $(\epsilon, \delta)$-DP if, for all subsets of possible outputs $\mathcal{S} \subseteq \text{Range}(\mathcal{M})$, and for all pairs of databases $D$ and $D'$ that differ by a single element, the following condition holds:
\begin{align}
    \operatorname{Pr}(\mathcal{M}(D) \in \mathcal{S}) \leq e^{\epsilon}  \operatorname{Pr}(\mathcal{M}(D') \in \mathcal{S}) + \delta,
\end{align}
where the smaller values of $\epsilon$ and $\delta$ indicate higher levels of privacy and vice versa. The essence of this definition is that when a client's data is added to or removed from the database, the distribution of outcomes from a private mechanism should remain largely unchanged \cite{Kalra2023, dwork2014algorithmic}. Under these circumstances, the SFL model owner cannot infer specific details about a client's data by observing the mechanism's output, thereby preserving privacy. From the study \cite{Lee2024a}, introducing Gaussian noise to the smashed data degraded the quality of the reconstructed image, thereby enhancing privacy. If the SFL model owner permits clients to utilize DP during SFL training, the owner can potentially set a larger $L_{\min}$ due to the reduced privacy concerns. In this context, we additionally perform the empirical study, shown in Table \ref{table:ssim}. To examine the relationship between noise levels and privacy leakage, we utilized a training-based adversarial inversion approach. The CIFAR-10 dataset was used for training an SFL model, a shadow model, and a reconstruction model. As we discussed, the privacy risk was evaluated by measuring the perceptual similarity between original and reconstructed images using the SSIM. From this result, clients must carefully determine the noise level in DP, balancing the trade-off between privacy and accuracy. Moreover, when the client complexity is low ($L_c$=3 in this case), the accuracy of SFL is more sensitive to the added noise. 
\begin{table}[]    
    \caption{Accuracy and SSIM with varying noise level ($\sigma$)}
    \centering
    \begin{tabular}{c|c|c|c|c}
    \hline\hline  
     & Cut layer & $\sigma=0$ & $\sigma=1$ & $\sigma=2$  
     \\ [0.5ex]
     \hline  
     \multirow{2}{*}{Accuracy} & $L_c=3$  & 90.8 & 90.1 & 88.0  \\
                               & $L_c=12$ & 92.1 & 91.8 & 92.0  \\
     \hline  
     \multirow{2}{*}{Privacy Leakage (SSIM)} & $L_c=3$  & 0.9563 & 0.7692 & 0.5194  \\
                           & $L_c=12$ & 0.2721 & 0.1606 & 0.0921  \\
    \hline 
    \end{tabular}
    \label{table:ssim} 
\end{table}

\subsection{Limitations}
The present work does not include a real-world deployment-level validation of the proposed game-theoretic mechanism for SFL. Instead, the proposed framework is validated through analytical characterization and dataset-based experiments. While this controlled setting is suitable for studying the SE behavior and the associated PoA, extending the proposed framework to real-world deployments is an important direction, as implementing a practical testbed would require incorporating system-level constraints and potentially redesigning parts of the analytical model, which constitutes a distinct research direction beyond the scope of this paper.

Another limitation of the present work is the simplified communication model adopted for analytical tractability. Specifically, we assume quasi-static uplink and downlink channels with fixed transmission rates and abstract the communication overhead using transmission latency derived from the data size. As a result, practical networking effects, including bandwidth contention, medium access scheduling, packet losses, wireless interference, and time-varying channel conditions, are not explicitly incorporated into the proposed framework. In addition, the current model can be extended to account for device heterogeneity by considering other practical factors, such as wireless transmission capability and battery constraints. Furthermore, the size of the smashed data is model-dependent and generally varies with the selected cut layer. However, deriving an analytical expression for the smashed-data size as a function of the cut layer is highly dependent on the underlying deep neural network architecture. Therefore, we approximate the smashed-data size using its average value across candidate cut layers to maintain analytical tractability in this initial study. Extending the proposed framework to incorporate more realistic wireless communication models, richer device heterogeneity, and layer-dependent smashed-data modeling for different neural network architectures constitutes an important direction for future work.

From the data perspective, another simplification of the proposed framework is that each client's contribution is represented solely by the amount of data contribution. Consequently, the current analytical model does not explicitly capture data quality or statistical heterogeneity (e.g., non-IID data distributions). While this abstraction enables closed-form equilibrium analysis, practical FL and SFL systems often exhibit heterogeneous datasets whose quality significantly affects learning performance. For instance, a client with fewer but more informative or statistically representative samples may contribute more to model convergence than another client possessing a much larger but highly redundant or biased dataset. In such scenarios, the utility of each client would depend not only on the amount of contributed data but also on its statistical quality. For example, representative data-quality measures such as the Earth Mover's Distance (EMD) \cite{Yuan2024}, and the Wasserstein distance \cite{Pang2023} could be incorporated into the utility function to quantify the contribution of heterogeneous local datasets. Consequently, the clients' best-response strategies and the resulting Nash and Stackelberg equilibria would jointly depend on both data quantity and data quality, leading to a more general equilibrium formulation. Extending the proposed framework in this direction constitutes an important avenue for future research.

Another limitation is that the cut-layer optimization considers privacy only through a predefined minimum cut-layer threshold. Consequently, privacy is not explicitly incorporated into the utility functions or the Stackelberg game. Furthermore, the present work adopts SSIM as an empirical proxy for reconstruction privacy following existing SFL studies, rather than an information-theoretic leakage metric. Extending the proposed framework to jointly optimize privacy, incentive, and cut-layer selection using rigorous privacy metrics and cumulative privacy leakage models over multiple communication rounds constitutes an important direction for future work.

\section{Conclusion and Future Work}
This paper has studied the interaction between an SFL model owner and clients during the SFL process, formulated as a single-leader multi-follower Stackelberg game in competitive or conflict situations. Each client's strategy involves the amount of data contributed for local training, while the SFL model owner determines suitable incentives to motivate client participation in SFL. The SFL model owner also sets the cut layer to balance the training burden between clients and servers while trying to meet the required privacy level. Our results indicate that Stackelberg strategies are desirable operating points for all participants in competitive scenarios, while achieving a well-balanced trade-off between model accuracy and the associated computing and networking overhead during the SFL process.  
For a more practical approach, future work may involve designing a more sophisticated incentive management framework that accounts for unshared information among clients, fine-grained evaluation of client contributions, and dynamic computing and networking conditions (e.g., via Stackelberg Bayesian games or deep reinforcement learning).

\bibliographystyle{IEEEtran}


\begin{IEEEbiography}[{\includegraphics[width=1in,height=1.25in,clip,keepaspectratio]{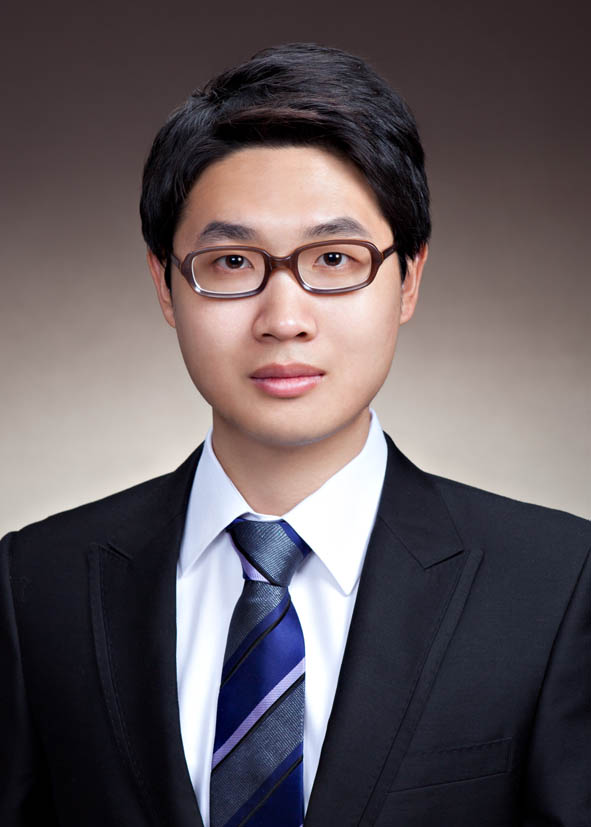}}]
{Joohyung Lee} (S'09, M'14, SM'19) is currently an Associate Professor in the School of Computing at Gachon University and a Visiting Fellow with the Department of Electrical and Computer Engineering at Princeton University. He received the B.S., M.S., and Ph.D. degrees from the Korea Advanced Institute of Science and Technology (KAIST), Daejeon, South Korea, in 2008, 2010, and 2014, respectively. From 2012 to 2013, he was a Visiting Researcher with the Information Engineering Group, Department of Electronic Engineering, City University of Hong Kong, Hong Kong. From 2014 to 2017, he was a Senior Engineer with Samsung Electronics. He has contributed several technical articles to the International Telecommunication Union Telecommunication Standardization Sector (ITU-T) and the 3rd Generation Partnership Project (3GPP). He has served as an Associate Editor of the IEEE Wireless Communications Letters. His current research interests include federated learning, 6G networks, cloud/edge computing, network economics, and AI-native wireless networks.
\end{IEEEbiography}

\begin{IEEEbiography}[{\includegraphics[width=1in,height=1.25in,clip,keepaspectratio]{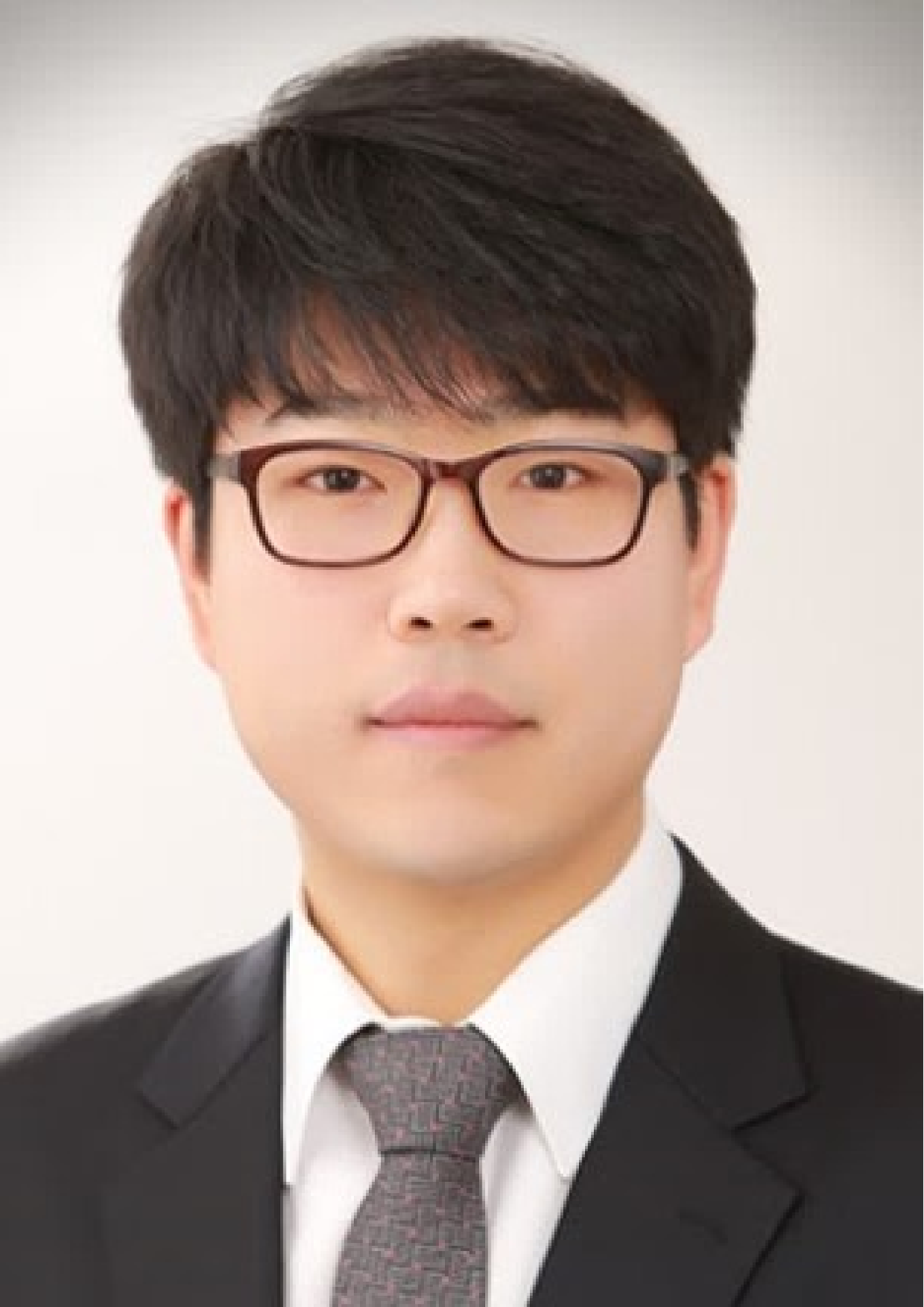}}]
{Jungchan Cho} received a BS degree from the School of Electrical and Electronics Engineering at Chung-Ang University, Seoul, South Korea, in 2010, and a Ph.D. from the Department of Electrical and Computer Engineering at Seoul National University, Seoul, South Korea, in 2016. He is currently an Associate Professor in the School of Computing at Gachon University, Seongnam, South Korea. From 2016 to 2019, he worked as a Senior Software Engineer at Samsung Electronics. From 2025 to 2026, he was a Visiting Scholar at the University of Birmingham, Birmingham, U.K. His research interests include federated learning, deep learning, computer vision, and machine learning.
\end{IEEEbiography}

\begin{IEEEbiography}[{\includegraphics[width=1in,height=1.25in,clip,keepaspectratio]{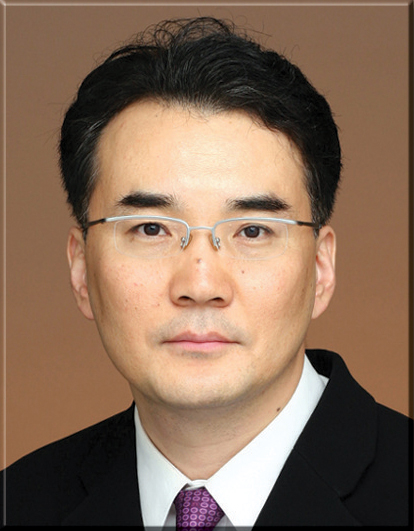}}]
{Wonjun Lee} (M’00-SM’06-F’21) received the B.S. and M.S. degrees in computer engineering from Seoul National University, Seoul, Republic of Korea, in 1989 and 1991, respectively, the M.S. degree in computer science from the University of Maryland, College Park, MD, USA, in 1996, and the Ph.D. degree in computer science and engineering from the University of Minnesota, Minneapolis, MN, USA, in 1999. He has been with the College of IT Convergence at Gachon University as a distinguished professor and the founding director of GAIC (Gachon Institute of AI and Computing) since September 2025. He was a professor at the School of Cybersecurity and the Department of Computer Science and Engineering at Korea University until August 2025. His research interests include communication and network protocols, wireless communication and networking optimization techniques, security and privacy in mobile computing, and RF-powered computing and networking. He has authored 15 international patents, over 250 papers in refereed international journals and conferences, and a book, ``Optimal Coverage in Wireless Sensor Networks," Springer, 2020. He has served on program and organization committees of numerous leading wireless and networking conferences, including IEEE INFOCOM from 2008 to 2025, General Chair of IEEE ICDCS 2026, PC Track Chair of IEEE ICDCS 2019, Workshop Chair of IEEE ICDCS 2023, and ACM MobiHoc from 2008 to 2009, as well as over 148 international conferences. He has received numerous awards, including the IEEE Chester W. Sall Memorial Award (2018), KIISE Gaheon Research Award (2011), LG Yonam Foundation Overseas Faculty Member Award (2007), Best Teaching Award (2005, 2009, 2021) from Korea University, and the Best Paper Awards from IEEE ICOIN 2002, ICOIN 2008, and IEEE SocialCom 2016. In 2019, his project, ``BackPlugged: Wearable-optimized Ultra Low-Power Wi-Fi Networking with Plugged-in Backscatter Radio," was selected for one of 100 Outstanding National R\&D Achievements by the Ministry of Science and ICT (MSIT) in Korea. He was a recipient of the Korean Government Overseas Scholarship between 1993 and 1996. He is the 2022 President-Elect (2023 President) of the Korean Institute of Information Scientists and Engineers (KIISE). He is a Fellow of the IEEE.

\end{IEEEbiography}

\begin{IEEEbiography}[{\includegraphics[width=1in,height=1.25 in,clip,keepaspectratio]{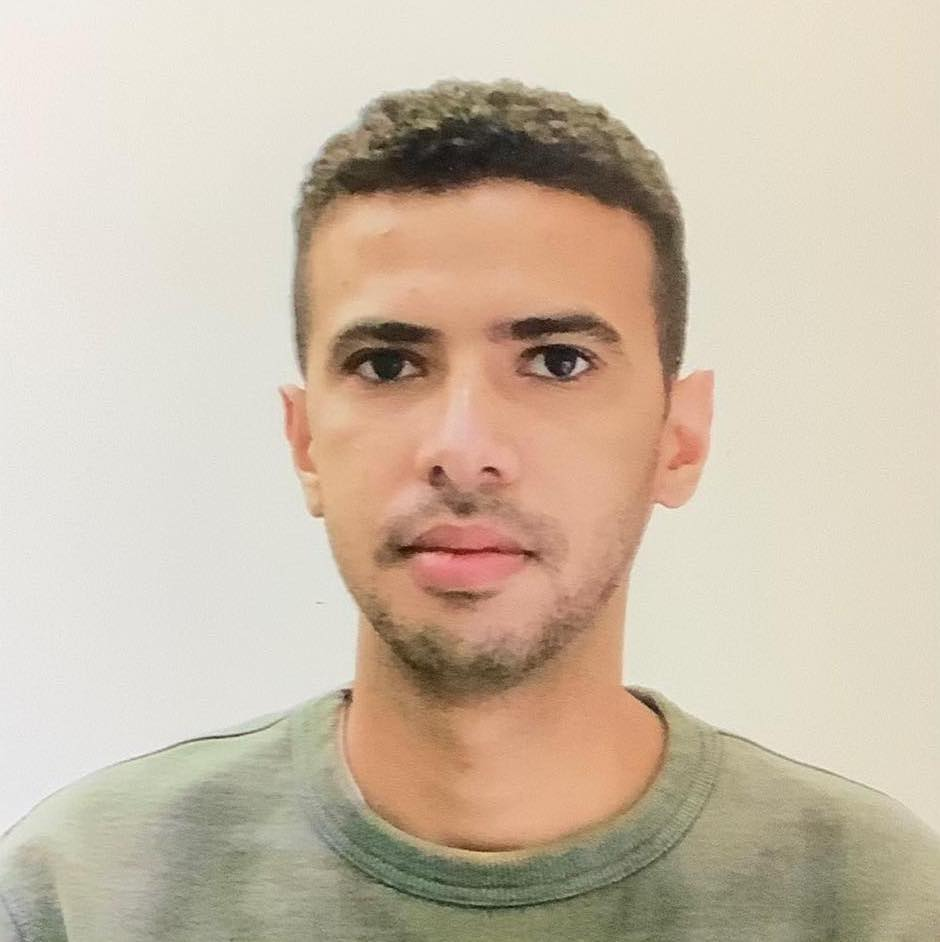}}]
{Mohamed Seif} (Member, IEEE) received the Ph.D. degree in Electrical and Computer Engineering from the University of Arizona, Tucson, AZ, USA. He is currently an Assistant Professor of Computer Science and Engineering at Oakland University, Rochester, MI, USA, and a Visiting Fellow at Princeton University, Princeton, NJ, USA. His research interests lie at the intersection of information and coding theory, machine learning, wireless communications, and data privacy, with a particular emphasis on developing secure, private, and trustworthy intelligent systems.
Dr. Seif has served on the Technical Program Committees of IEEE ICC, IEEE GLOBECOM, and IEEE ICNC. He also served as a Guest Editor for the Entropy Special Issue on information-theoretic approaches to security and privacy. He is an inventor on two pending U.S. patent applications.
\end{IEEEbiography}

\begin{IEEEbiography}[{\includegraphics[width=1in,height=1.25in,clip,keepaspectratio]{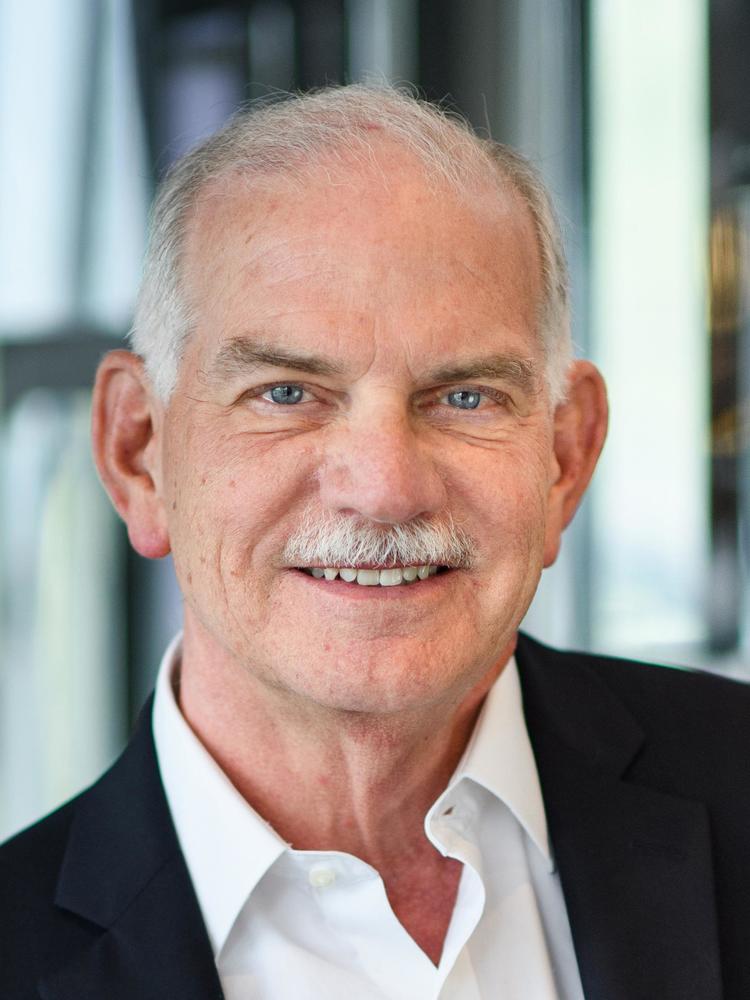}}]
{H. Vincent Poor} (S'72, M'77, SM'82, F'87) received the Ph.D. degree in EECS from Princeton University in 1977. From 1977 until 1990, he was on the faculty of the University of Illinois at Urbana-Champaign. Since 1990 he has been on the faculty at
Princeton, where he is currently the Michael Henry Strater University Professor. During 2006 to 2016, he served as the dean of Princeton’s School of Engineering and Applied Science, and he has also held visiting appointments at several other universities, including most recently at Berkeley and Caltech. His research interests are in the areas of information theory, stochastic analysis and machine learning, and their applications in wireless networks, energy systems and related fields. Among his publications in these areas is the book \textit{Machine Learning and Wireless Communications}. (Cambridge University Press, 2022). Dr. Poor is a member of the National Academy of Engineering and the National Academy of Sciences and is a foreign member of the Royal Society and other national and international academies.
\end{IEEEbiography}

\end{document}